\documentclass[letterpaper,twocolumn,twoside]{IEEEtran}
\usepackage{amsmath, amssymb, bm, cite, epsfig, psfrag}
\usepackage{eufrak}
\usepackage{algorithm,algorithmic}
\usepackage{multirow}
\usepackage{array}

\def\beq{\begin{equation}}
\def\eeq{\end{equation}}
\def\beqa{\begin{eqnarray}}
\def\eeqa{\end{eqnarray}}
\def\beqan{\begin{eqnarray*}}
\def\eeqan{\end{eqnarray*}}

\def\R{{\mathbb{R}}}

\def\argmax{\mathop{\mathrm{arg\,max}}}

\def\x{\times}

\newtheorem{lemma}{Lemma}

\newtheorem{claim}{Claim}

\def\phat{\widehat{p}}

\def\rhat{\widehat{r}}
\def\shat{\widehat{s}}

\def\xhat{\widehat{x}}
\def\zhat{\widehat{z}}
\def\la{\leftarrow}
\def\ra{\rightarrow}

\def\arr{\rightarrow}
\def\var{\mbox{var}}
\def\cov{\mbox{cov}}

\def\tm1{t\! - \! 1}
\def\tp1{t\! + \! 1}

\def\phat{\widehat{p}}

\def\rhat{\widehat{r}}
\def\shat{\widehat{s}}

\def\xhat{\widehat{x}}

\def\zhat{\widehat{z}}

\def\taubar{\overline{\tau}}

\def\la{\leftarrow}
\def\ra{\rightarrow}

\def\arr{\rightarrow}
\def\Exp{\mathbb{E}}
\def\var{\mbox{var}}
\def\tm1{t\! - \! 1}
\def\tp1{t\! + \! 1}

\def\abf{\mathbf{a}}
\def\bbf{\mathbf{b}}
\def\dbf{\mathbf{d}}

\def\pbfhat{\widehat{\mathbf{p}}}
\def\qbf{\mathbf{q}}

\def\rbfhat{\widehat{\mathbf{r}}}

\def\ubf{\mathbf{u}}

\def\vbf{\mathbf{v}}

\def\wbf{\mathbf{w}}

\def\xbf{\mathbf{x}}
\def\xbfhat{\widehat{\mathbf{x}}}

\def\ybf{\mathbf{y}}
\def\zbf{\mathbf{z}}

\def\zbfhat{\widehat{\mathbf{z}}}
\def\Abf{\mathbf{A}}
\def\Bbf{\mathbf{B}}

\def\Dbf{\mathbf{D}}
\def\Gbf{\mathbf{G}}
\def\Kbf{\mathbf{K}}

\def\Phat{\widehat{P}}
\def\Qbf{\mathbf{Q}}

\def\Rhat{\widehat{R}}

\def\Ubf{\mathbf{U}}
\def\Vbf{\mathbf{V}}

\def\Xhat{\widehat{X}}
\def\Xbf{\mathbf{X}}
\def\Ybf{\mathbf{Y}}
\def\Zbf{\mathbf{Z}}

\def\alphabf{{\boldsymbol \alpha}}
\def\betabf{{\boldsymbol \beta}}

\def\taubf{{\boldsymbol \tau}}
\def\lambdabf{{\boldsymbol \lambda}}

\def\xibf{{\boldsymbol \xi}}

\def\Ggothic{\mathfrak{G}}

\def\fin{f_{\rm in}}
\def\fout{f_{\rm out}}
\def\gin{g_{\rm in}}
\def\gout{g_{\rm out}}
\def\Fin{F_{\rm in}}
\def\Fout{F_{\rm out}}
\def\Gin{G_{\rm in}}
\def\Gout{G_{\rm out}}
\def\ppost{p^{\rm post}}

\title{Generalized Approximate Message Passing for
Estimation with Random Linear Mixing}
\author{Sundeep Rangan, ~\IEEEmembership{Member,~IEEE}
\thanks{This material is based upon work supported by the National Science
Foundation under Grant No. 1116589.}
\thanks{S. Rangan (email: srangan@poly.edu) is with
        Polytechnic Institute of New York University, Brooklyn, NY.}
}

\begin{document}


\setcounter{page}{1}

\maketitle
\begin{abstract}
We consider the estimation of an i.i.d.\ random vector
observed through a linear transform followed by
a componentwise, probabilistic (possibly nonlinear) measurement channel.
A novel algorithm, called generalized approximate message passing (GAMP),
is presented that provides computationally efficient
approximate implementations of max-sum and sum-problem loopy belief propagation
for such problems.  
The algorithm extends earlier approximate message passing methods to incorporate
arbitrary distributions on both the input and output of the transform
and can be applied to a wide range of problems in nonlinear compressed sensing and
learning.

Extending an analysis by Bayati and Montanari, we argue that
the asymptotic componentwise behavior of the GAMP method under large, i.i.d.\ Gaussian
transforms is described by a simple set of state evolution (SE) equations.
From the SE equations, one can \emph{exactly} predict the asymptotic value of
virtually any componentwise
performance metric including mean-squared error or detection accuracy.
Moreover, the analysis is valid for arbitrary input and output distributions, even when
the corresponding optimization problems are non-convex.
The results match predictions by Guo and Wang for relaxed belief propagation on
large sparse matrices
and, in certain instances, also agree with the optimal performance predicted by the
replica method.
The GAMP methodology thus provides a computationally efficient methodology,
applicable to a large class of non-Gaussian estimation problems with precise
asymptotic performance guarantees.
\end{abstract}



\section{Introduction} \label{sec:intro}

The problem of estimating vectors from linear transforms followed by random,
possibly nonlinear, measurements,
arises in a range of problems in signal processing, communications,
 and learning.
This paper considers a general class of such estimation problems in a Bayesian setting
shown in Fig.\ \ref{fig:AMPGenChan}.
An input vector $\qbf \in Q^n$ has components $q_j \in Q$
for some set $Q$ and
generates an unknown random vector $\xbf \in \R^n$ through
a componentwise input channel described by a conditional distribution
$p_{X|Q}(x_j|q_j)$.
The vector $\xbf$ is then passed through a linear transform
\beq \label{eq:zAx}
    \zbf = \Abf \xbf,
\eeq
where $\Abf \in \R^{m \x n}$ is a known transform matrix.
Finally, each component of $z_i$ of $\zbf$ randomly generates
an output component $y_i$ of a vector $\ybf \in Y^m$ through a second
scalar conditional distribution $p_{Y|Z}(y_i|z_i)$,
where $Y$ is some output set.
The problem is to estimate the transform input $\xbf$
and output $\zbf$ from the system input
vector $\qbf$, system output $\ybf$ and linear transform $\Abf$.

The formulation is general and has wide applicability -- we will present
several applications in Section \ref{sec:examples}.
However, for many non-Gaussian instances of the problem, exact computations of
quantities such as the posterior mode or mean  of $\xbf$ is computationally prohibitive.
The primary difficulty is that the matrix $\Abf$
``couples" or ``mixes" the coefficients of $\xbf$ into $\zbf$.
If the transform matrix
were the identity matrix (i.e.\ $m=n$ and $\Abf = I$), then
the estimation problem would \emph{decouple}
into $m=n$ scalar estimation problems, each defined by the
Markov chain:
\[
    q_j \sim p_Q(q_j) \stackrel{p_{X|Q}(x_j|q_j)}{\longrightarrow} x_j = z_j
        \stackrel{p_{Y|Z}(y_j|z_j)}{\longrightarrow} y_j.
\]
Since all the quantities in this Markov chain are scalars, one could
numerically compute
the exact posterior distribution on any component
$x_j (= z_j)$ from $q_j$ and $y_j$ with one-dimensional integrals.
However, for a general (i.e.\ non-identity) matrix $\Abf$,
the components of $\xbf$ are coupled into the vector $\zbf$.
In this case, the posterior distribution of any particular component $x_j$
or $z_i$ would involve a high-dimensional integral that is, in general, difficult
to evaluate.

This paper presents a novel algorithm,
\emph{generalized approximate message passing} or GAMP,
that \emph{decouples} the vector-valued estimation problem into a
sequence of scalar problems and linear transforms.
The GAMP algorithm is an extension of several earlier
Gaussian and quadratic approximations of loopy belief propagation (loopy BP)
that have been successful in many previous applications,
including, most recently, compressing sensing
\cite{BoutrosC:02,NeirottiS:05,TanakaO:05,GuoW:06,GuoW:07,DonohoMM:09,BayatiM:11,Rangan:10arXiv,Montanari:11}
(See the ``Previous Works" subsection below).
The proposed methodology extends these methods to provide
several valuable features:
\begin{itemize}
\item \emph{Generality:}  Most importantly, the
GAMP methodology applies to essentially arbitrary
priors $P_{X|Q}$ and output channel distributions $P_{Y|Z}$.  The algorithm can thus
incorporate arbitrary non-Gaussian inputs as well as output nonlinearities.
Moreover, the algorithm can be used to realize approximations of both max-sum
loopy BP for \emph{maximum a posteriori} (MAP) estimation and sum-product loopy BP
for minimum mean-squared error (MMSE) estimates and approximations of the
posterior marginals.

\item \emph{Computational simplicity:}  The algorithm is computationally simple
with each iteration involving only scalar estimation computations on the
components of $\xbf$ and $\zbf$ along with linear transforms.
Indeed, the iterations are similar in form to the fastest known methods
for such problems
 \cite{WrightNF:09,ForGlow:83,GlowLeTal:89,HeLiaoHanY:02,WenGolYin:10},
 while being potentially more general.  Moreover, our simulations indicate excellent
 convergence in a small number, typically 10 to 20, iterations.

\item \emph{Analytic tractability:}  Our main theoretical result, Claim~\ref{thm:gampSE},
shows that for large Gaussian i.i.d.\ transforms, $\Abf$,
the asymptotic behavior of the components of the GAMP algorithm
are described exactly by a simple \emph{scalar equivalent model}.
The parameters in this model can be computed
by a set of scalar \emph{state evolution} (SE) equations, analogous to
the density evolution equations for BP decoding of LDPC codes
\cite{TenBrink:01,AshkiminKT:04}.  From this scalar equivalent model,
one can asymptotically \emph{exactly} predict any componentwise performance metric
such as mean-squared error (MSE) or detection accuracy.
Moreover, the SE analysis generalizes results on earlier approximate
message passing-like algorithms
\cite{BoutrosC:02,NeirottiS:05,TanakaO:05,GuoW:06,GuoW:07,DonohoMM:09,BayatiM:11,Rangan:10arXiv,Montanari:11},
and also, in certain instances, show a match with the optimal performance
as predicted by the replica method in statistical physics
\cite{Tanaka:02,GuoV:05,Rangan:10arXiv,KabashimaWT:09arXiv,KrzMSSZ:11-arxiv,KrzMSSZ:12-arxiv}.
\end{itemize}

The GAMP algorithm thus provides a general and systematic approach
to a large class of estimation problems with linear mixing
that is computationally tractable and widely-applicable.
Moreover, in the case of large random $\Abf$, the method admits
exact asymptotic performance characterizations.

\subsection{Prior Work} \label{sec:priorWork}

The GAMP algorithm belongs to a long line of methods
based on Gaussian and quadratic 
approximations of loopy belief propagation (loopy BP).
Loopy BP is a general methodology for estimation problems
where the statistical relationships between variables are represented
via a graphical model.  The loopy BP algorithm iteratively updates
estimates of the
variables via a message passing procedure along the graph edges
\cite{Pearl:88,WainwrightJ:08}.
For the linear mixing estimation problem, the graph is precisely
the incidence matrix of the matrix $\Abf$
and the loopy BP messages are passed between nodes corresponding to the input
variables $x_j$ and outputs variables $z_i$.

The GAMP algorithm provides approximations of two
important variants of loopy BP:
\begin{itemize}
\item \emph{Max-sum} loopy BP
for approximately computing MAP estimates of $\xbf$ and $\zbf$
as well as sensitivities of the maximum of the  posterior distribution
to the individual components of $\xbf$ and $\zbf$; and
\item \emph{Sum-product} loopy BP for approximations of the
minimum mean-squared error (MMSE) estimates of $\xbf$ and $\zbf$
and the marginal posteriors of their individual components.
\end{itemize}
We will call the GAMP approximations of the two algorithms,
max-sum GAMP and sum-product GAMP, respectively.

In communications and signal processing,
BP is best known for its connections to
iterative decoding in turbo and LDPC codes \cite{McElieceMC:98,MacKay:99,MacKayN:97}.
However, while turbo and LDPC codes typically
involve computations over finite fields,
BP has also been successfully applied in a number of
problems with linear real-valued mixing, including
CDMA multiuser detection \cite{BoutrosC:02,YoshidaT:06,GuoW:08},
lattice codes \cite{SommerFS:08} and
compressed sensing \cite{BaronSB:10,GuoBS:09-Allerton,DonohoMM:09}.

\begin{figure}
\begin{center}
  \includegraphics[width=3.2in,height=1.4in]{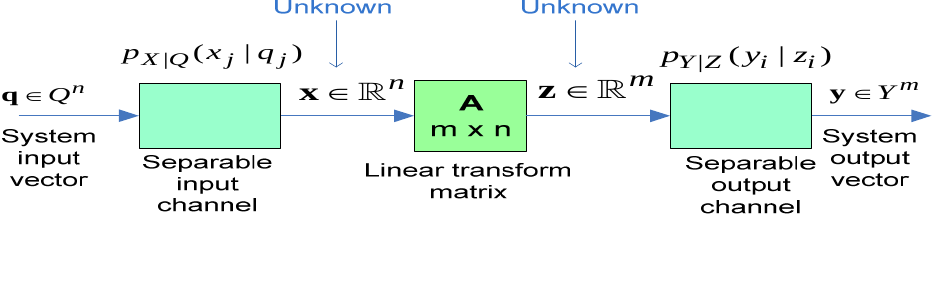}
\end{center}
\caption{General estimation problem with linear mixing.  The problem
is to estimate the transform input and output vectors $\xbf$ and $\zbf$
from the system input and output vectors $\qbf$ and $\ybf$,
channel transition probability functions
$p_{X|Q}(\cdot|\cdot)$, $p_{Y|Z}(\cdot|\cdot)$ and transform matrix $\Abf$.}
\label{fig:AMPGenChan}
\end{figure}

Although exact implementation of loopy BP for dense graphs is generally
computationally difficult,
Gaussian and quadratic approximations of loopy BP have been widely and successfully
used for many years.
Such approximate methods originated in CDMA multiuser detection problems in
\cite{BoutrosC:02,NeirottiS:05,TanakaO:05,GuoW:06,GuoW:07},
and, more recently, have attracted considerable interest in compressed sensing \cite{DonohoMM:09,BayatiM:11,Rangan:10arXiv,Montanari:11,KrzMSSZ:11-arxiv,KrzMSSZ:12-arxiv}.
The methods have appeared under a variety of names including ``approximate BP",
``relaxed BP" and, most recently, ``approximate message passing".
Gaussian approximations are also used in expectation propagation
\cite{Minka:01,Seeger:08}, as well as the analysis and design of
LDPC codes \cite{ChungRU:01,ElGamalH:01,Varshney:07}.

The GAMP algorithm presented here is most closely related to the
approximate message passing (AMP) methods in
\cite{DonohoMM:09,DonohoMM:10-ITW1,DonohoMM:10-ITW2}
as well as the relaxed BP methods in \cite{GuoW:06,GuoW:07,Rangan:10arXiv}.
The AMP method \cite{DonohoMM:09} used a quadratic approximation of
max-sum loopy BP to derive an efficient implementation of the LASSO estimator
\cite{Tibshirani:96,ChenDS:96}.  The LASSO estimate
is equivalent to a MAP estimate of $\xbf$ assuming a Laplacian prior.
This method was then generalized in \cite{DonohoMM:10-ITW1,DonohoMM:10-ITW2}
to implement Bayesian estimation with arbitrary priors
using Gaussian approximations of sum-product loopy BP.
The relaxed BP method in \cite{GuoW:07,Rangan:10arXiv} offers a further
extension to nonlinear output channels, but the algorithm is limited to
sum-product approximations of loopy BP and the analysis is limited
to certain random, sparse matrices.

The GAMP method proposed in this paper provides a unified
methodology for estimation with dense matrices that can incorporate
arbitrary input and output distributions and provide
approximations of both max-sum and sum-product loopy BP.

Moreover, similar to previous analyses of AMP-like algorithms,
we argue that that the asymptotic behavior of GAMP has a sharp characterization for
certain large, i.i.d.\ matrices.
Specifically, for large random $\Abf$, the componentwise behavior of many AMP techniques
can be asymptotically described exactly by a set of state evolution (SE) equations.
Such analyses have been presented  for both large sparse matrices
\cite{BoutrosC:02,GuoW:06,GuoW:07,Rangan:10arXiv},
and, more recently, for dense i.i.d.\ matrices \cite{DonohoMM:09,BayatiM:11}.
The validity of the SE analysis on dense matrices was particularly difficult
to rigorously prove since conventional density evolution techniques
such as \cite{TenBrink:01,AshkiminKT:04} need graphs that are locally cycle-free.
The key breakthrough was an analysis by Bayati and Montanari in \cite{BayatiM:11}
that provided the first completely rigorous analysis of the dynamics of message
passing in dense graphs using a new conditioning argument from
\cite{Bolthausen:09}.
That work also provided a rigorous justification of many predictions
\cite{Tanaka:02,GuoV:05,KabashimaWT:09arXiv,RanganFG:12-IT,KrzMSSZ:11-arxiv,KrzMSSZ:12-arxiv}
of the behavior of optimal estimators using the replica method from statistical physics.

The main theoretical contribution of this paper, Claim~\ref{thm:gampSE},
can be seen as 
an extension of Bayati and Montanari's
analysis in \cite{BayatiM:11} to GAMP -- the novelty being the
incorporation of arbitrary output channels.
Specifically, we present a generalization of 
the SE equations to arbitrary output channels recovering several earlier
results on AWGN channels as special instances.  The SE equations also agree with the results in
\cite{GuoW:07,Rangan:10arXiv} for relaxed BP applied to
arbitrary output channels for large sparse matrices.

The extension of Bayati and Montanari's
analysis in \cite{BayatiM:11} to incorporate arbitrary output channels is relatively 
straightforward.  Unfortunately, a full re-derivation of their result would be long and 
beyond the scope of this paper.  We thus only provide a brief of sketch of the main
steps and our result is thus not fully rigorous.  To emphasize the lack of rigor, we 
use the term Claim instead of Theorem.  The predictions, however, are confirmed in 
numerical simulations.

A conference version of this paper appeared in \cite{Rangan:11-ISIT}. This paper
includes all the proofs and derivations along with more detailed discussions and
simulations.

\subsection{Outline}
The outline of the remainder of this paper is as follows.
As motivation, Section~\ref{sec:examples} provides some examples of the linear mixing problems.
The GAMP method is then introduced in Section~\ref{sec:GAMP},
and precise equations for the max-sum and sum-product
versions are given in Section~\ref{sec:bpApprox}.  
The asymptotic state evolution analysis is presented in Section~\ref{sec:asymAnal}
and shown to recover many previous predictions of the SE analysis
in several special cases discussed in Section~\ref{sec:special}.
A simple numerical simulation to confirm the predictions is presented in Section~\ref{sec:sparseNL} which considers a nonlinear compressed sensing problem.
Finally, since the original publication of the paper, there has been 
considerable work on the topic.  The conclusions summarize this work and present 
avenues for future research.

\section{Examples and Applications} \label{sec:examples}
The linear mixing model
is extremely general and can be applied in a range of circumstances.
We review some simple examples for both the output and input
channels from \cite{Rangan:10arXiv}.

\subsection{Output Channel Examples} \label{sec:outChanEx}

\paragraph{AWGN output channel}  For an additive white Gaussian noise (AWGN)
output channel, the output vector $\ybf$ can be written as
\beq \label{eq:yzw}
    \ybf = \zbf + \wbf = \Abf\xbf + \wbf,
\eeq
where $\wbf$ is a zero mean, Gaussian i.i.d.\ random vector
independent of $\xbf$.
The corresponding channel transition probability
distribution is given by
\beq \label{eq:pyzAwgn}
    p_{Y|Z}(y|z) = \frac{1}{\sqrt{2\pi\tau^w}}
    \exp\left(-\frac{(y-z)^2}{2\tau^w}\right),
\eeq
where $\tau^w > 0$ is the variance of the components of $\wbf$.

\paragraph{Non-Gaussian noise models}
Since the output channel can incorporate an arbitrary
separable distribution,
the linear mixing model can also include
the model (\ref{eq:yzw}) with non-Gaussian
noise vectors $\wbf$, provided the components of
$\wbf$ are i.i.d.
One interesting application for a non-Gaussian noise model
is to study the bounded noise that arises in quantization
as discussed in \cite{Rangan:10arXiv}.

\paragraph{Logistic channels}  A
quite different
channel is based on a \emph{logistic} output.
In this model, each output $y_i$ is $0$ or $1$,
where the probability
that $y_i=1$ is given by some sigmoidal function such as
\beq \label{eq:logitOut}
    p_{Y|Z}(y_i=1|z_i) = \frac{1}{1+a\exp(-z_i)},
\eeq
for some constant $a > 0$.  Thus, larger values of $z_i$
result in a higher probability that $y_i=1$.

This logistic model can be used in classification problems as follows \cite{Bishop:06}:
Suppose one is given $m$ samples, with each sample being labeled
as belonging to one of two classes.
Let $y_i = 0$ or 1 denote the class of sample $i$.
Also, suppose that the $i$th row of the transform matrix $\Abf$
contains a vector of $n$ data values associated with the $i$th sample.
Then, using a logistic channel model such as (\ref{eq:logitOut}),
the problem of estimating the vector $\xbf$ from the labels $\ybf$
and data $\Abf$ is equivalent to finding a linear
dependence on the data that classifies the samples between the
two classes.
This problem is often referred to as logistic regression
and the resulting vector $\xbf$ is called the regression vector.
By adjusting the prior on the components of $\xbf$, one
can then impose constraints on the components of $\xbf$
including, for example, sparsity constraints.

\subsection{Input Channel Examples} \label{sec:inChanEx}
The distributions $p_{X|Q}(x_j|q_j)$ models the prior on $x_j$,
with the variable $q_j$ being a parameter of that prior that varies
over the components, but is known to the estimator.
When all the components of $x_j$ are identically distributed, we
can ignore $q_j$ and use a constant distribution $p_X(x_j)$.

\paragraph{Sparse priors and compressed sensing}
One class of non-Gaussian priors that have attracted considerable
recent attention is sparse distributions.
A vector $\xbf$ is sparse if a large fraction of its components are zero
or close to zero.  Sparsity can be modeled probabilistically with a variety
of heavy-tailed distributions $p_X(x_j)$
including Gaussian mixture models,
generalized Gaussians and Bernoulli distributions with a high probability of
the component being zero.
The estimation of sparse vectors with random linear measurements is
the basic subject of compressed sensing \cite{CandesRT:06-IT,Donoho:06,CandesT:06}
and fits naturally into the linear mixing framework.

\paragraph{Discrete distributions}
The linear mixing model can also incorporate discrete distributions on the
components of $\xbf$.
Discrete distribution arise often in communications problems
where discrete messages are modulated onto the components of $\xbf$.
The linear mixing with the transform matrix $\Abf$ comes into play
in CDMA spread spectrum systems and lattice codes mentioned above.

\paragraph{Power variations and dynamic range}
For multiuser detection problems in communications,
the parameter $q_j$ can be used to model \emph{a priori}
power variations amongst the users that are known to the receiver.
For example, suppose $q_j > 0$ almost surely and
$x_j = \sqrt{q_j}u_j$, with $u_j$ being identically
distributed across all components $j$ and independent of $q_j$.
Then, $\Exp(x_j^2|q_j) = q_j\Exp(u_j^2)$ and
$q_j$ thus represents a power scaling.
This model has been used in a variety of analyses of CDMA multiuser
detection \cite{TseH:99,GuoV:05,GuoW:06}
and random access channels \cite{FletcherRG:09arXiv}.

\section{Generalized Approximate Message Passing} \label{sec:GAMP}

We begin with a description of the generalized approximate message passing
(GAMP) algorithm, which is an extension of the AMP procedure in \cite{DonohoMM:09,BayatiM:11}.
Similar to the AMP method, the idea of the algorithm is to iteratively decompose
the vector valued estimation problem into a sequence of
scalar operations at the input and output nodes.
In the GAMP algorithm, the scalar operations are defined by two functions,
$\gout(\cdot)$ and $\gin(\cdot)$, that we call the \emph{scalar estimation} functions.
We will see in Section \ref{sec:bpApprox} that with appropriate
choices of these functions, the GAMP algorithm can provide Gaussian and quadratic
approximations of either sum-product and max-sum loopy BP.

The steps of the GAMP method are shown in Algorithm \ref{algo:GAMP}.
The algorithm produces a sequence of estimates, $\xbfhat(t)$ and $\zbfhat(t)$,
for the unknown vectors $\xbf$ and $\zbf$.  The algorithm also outputs vectors
$\taubf^x(t)$ and $\taubf^s(t)$.  As we will see in the case of the sum-product
GAMP (Section \ref{sec:bpApproxMmse}), these have interpretations as certain variances.
Although our analysis later is for real-valued matrices, we have written the
equations for the complex case. 

\begin{algorithm}
\caption{Generalized AMP (GAMP)}
\label{algo:GAMP}
Given a matrix $\Abf \in \R^{m \x n}$, system inputs and outputs $\qbf$  and
$\ybf$ and scalar estimation functions $\gin(\cdot)$, and $\gin(\cdot)$,
generate a sequence of estimates
$\xbfhat(t)$, $\zbfhat(t)$, for $t=0,1,\ldots$ through the following recursions:
\begin{enumerate}
\item \emph{Initialization:}  Set $t=0$ and set $\xhat_j(t)$ and $\tau^x_j(t)$
to some initial values.
\item \emph{Output linear step:}
\begin{subequations} \label{eq:outLin}  For each $i$, compute:
\beqa
    \tau^p_i(t) &=& \sum_j |a_{ij}|^2\tau^x_j(t) \label{eq:mup} \\
    \phat_i(t) &=& \sum_j a_{ij}\xhat_j(t) - \tau^p_i(t)\shat_i(\tm1),
        \label{eq:phati} \\
    \zhat_i(t) &=& \sum_j a_{ij}\xhat_j(t) \label{eq:zhati}
\eeqa
\end{subequations}
where initially, we take $\shat(-1) = 0$.

\item \emph{Output nonlinear step:}  For each $i$,
\begin{subequations} \label{eq:outNL}
\beqa
    \shat_i(t) &=& \gout(t,\phat_i(t),y_i,\tau^p_i(t)) \\
    \tau^s_i(t) &=& -\frac{\partial}{\partial \phat}
        \gout(t,\phat_i(t),y_i,\tau^p_i(t)).
\eeqa
\end{subequations}

\item \emph{Input linear step:}  For each $j$,
\begin{subequations} \label{eq:inLin}
\beqa
    \tau^r_j(t) &=& \left[ \sum_i |a_{ij}|^2\tau^s_i(t) \right]^{-1} \label{eq:mur}\\
    \rhat_j(t) &=& \xhat_j(t) + \tau^r_j(t)\sum_i a_{ij}\shat_i(t).
        \label{eq:rhat}
\eeqa
\end{subequations}

\item \emph{Input nonlinear step:}  For each $j$,
\begin{subequations} \label{eq:inNL}
\beqa
    \xhat_j(\tp1) &=& \gin(t,\rhat_j(t),q_j,\tau^r_j(t)) \\
    \tau^x_j(\tp1) &=& \tau^r_j(t)\frac{\partial}{\partial \rhat}
        \gin(t,\rhat_j(t),q_j,\tau^r_j(t)). \label{eq:mux}
\eeqa
\end{subequations}
Then increment $t=t+1$ and return to step 2 until a sufficient number
of iterations have been performed.
\end{enumerate}
\end{algorithm}

\subsection{Computational Complexity}
We will discuss the selection of the scalar estimation functions
and provide a detailed analysis of the algorithm performance later.
We first describe the algorithm's computational simplicity
-- which is a key part of the algorithm appeal.

Each iteration of the algorithm has four steps.  The first step,
the output linear step, involves only a matrix-vector multiplications by $\Abf$
and $|\Abf|^2$, where the squared magnitude is taken componentwise.
The worst case complexity is $O(mn)$ and would be smaller for structured transforms
such as Fourier, wavelet or sparse.
The next step --- the nonlinear step --- involves componentwise applications of the
nonlinear output estimation function $\gout(\cdot)$ on each of the
$m$ components of the output vector $\pbfhat$.  As we will see,
the function $\gout(\cdot)$ does not change with the dimension, so the total complexity
of the output nonlinear step is complexity of $O(m)$.
Similarly, the input steps involve matrix-vector multiplications by $\Abf^T$ and $(|\Abf|^2)^T$
along with componentwise scalar operations at the input.
The complexity is again dominated by the transforms with a worst-case complexity of
$O(mn)$.

Thus, we see the GAMP algorithm reduces the vector-valued operation to a sequence of
linear transforms and scalar estimation functions.  The worst-case total complexity per iteration
is thus $O(mn)$ and smaller for structured transforms.  Moreover, as we will see in
the state evolution analysis, the number of iterations required for the same per component
performance does not increase with the problem size, so the total complexity is bounded
by the matrix-vector multiplication.  In addition, we will see in the simulations that
good performance can be obtained with a small number of iterations, usually 10 to 20.

It should be pointed out that the structure of the GAMP iterations -- transforms
followed by scalar operations -- underly many of the most of successful methods
for linear-mixing type problems.  In particular, the separable approximation
method of \cite{WrightNF:09} and alternating direction methods in
\cite{ForGlow:83,GlowLeTal:89,HeLiaoHanY:02,WenGolYin:10} are all based on iterations
of this form.
The contribution of the current paper is to show that specific instance of these
transform + separating algorithms can be interpreted as an approximation of loopy BP
and admits precise asymptotic analyses.

\subsection{Further Simplifications}

To further reduce computational complexity,
the GAMP algorithm can be approximated by a modified procedure shown in
Algorithm \ref{algo:GAMP-simp}.
In the modified procedure, the variance vectors, $\taubf^r(t)$ and
$\taubf^p(t)$, are replaced with scalars $\tau^r(t)$ and $\tau^p(t)$,
thus forcing all the variance components to be the same.
This approximation would be valid when the components of the transform
matrix $\Abf$ are the approximately equal so that 
$|a_{ij}|^2 \approx \|\Abf\|^2_F/mn$ for all $i$, $j$.  

The approximations in  Algorithm~\ref{algo:GAMP-simp} can also be heuristically justified  
when $\Abf$ has i.i.d.\ components.  
Specifically, if $\Abf$ is i.i.d., $m$ and $n$ are large, and the dependence 
between the components of $|A_{ij}|^2$
and the vectors $\taubf^x(t)$ and $\taubf^s(t)$ can be ignored, the simplification in 
Algorithm~\ref{algo:GAMP-simp} can be justified through the Law of Large Numbers.
Of course, $|A_{ij}|^2$ is not independent of $\taubf^x(t)$ and $\taubf^s(t)$,
but in our simulation below for an i.i.d.\ matrix, 
we will see very little difference between 
Algorithm~\ref{algo:GAMP} and the simplified version, Algorithm~\ref{algo:GAMP-simp}.

The benefit of the simplification is that
we can eliminate the matrix multiplications by $|\Abf|^2$ and $(|\Abf|^2)^T$
-- reducing the number of transforms per iteration from four to two.
This savings can be significant, particularly when the $\Abf^2$ and $(\Abf^2)^T$
have no particularly simple implementation.  
The simplified algorithm, Algorithm~\ref{algo:GAMP-simp}, also more closely 
matches the AMP procedure of \cite{DonohoMM:09}, and will be more amenable to analysis later
in Section~\ref{sec:asymAnal}.

\begin{algorithm}
\caption{GAMP with Scalar Variances}
\label{algo:GAMP-simp}
Given a matrix $\Abf \in \R^{m \x n}$, a system input and output vectors $\qbf$ and $\ybf$, selection functions $\gin(\cdot)$, and $\gin(\cdot)$,
generate a sequence of estimates
$\xbfhat(t)$, $\zbfhat(t)$, for $t=0,1,\ldots$ through the recursions:
\begin{enumerate}
\item \emph{Initialization:}  Set $t=0$ and let $\xhat_j(t)$ and $\tau^x(t)$
be any initial sequences.
\item \emph{Output linear step:}
\begin{subequations} \label{eq:outLinSimp}  For each $i$, compute:
\beqa
    \tau^p(t) &=& (1/m)\|\Abf\|^2_F\tau^x(t) \label{eq:mupSimp} \\
    \phat_i(t) &=& \sum_j a_{ij}\xhat_j(t) - \tau^p(t)\shat_i(\tm1),
        \label{eq:phatiSimp} \\
    \zhat_i(t) &=& \sum_j a_{ij}\xhat_j(t) \label{eq:zhatiSimp}
\eeqa
\end{subequations}
where initially, we take $\shat(-1) = 0$.

\item \emph{Output nonlinear step:}  For each $i$,
\begin{subequations} \label{eq:outNLSimp}
\beqa
    \shat_i(t) &=& \gout(t,\phat_i(t),y_i,\tau^p(t)) \\
    \tau^s(t) &=& -\frac{1}{m}\sum_{i=1}^m \frac{\partial}{\partial \phat}
        \gout(t,\phat_i(t),y_i,\tau^p(t)). \hspace{0.5cm} \label{eq:tausSimp}
\eeqa
\end{subequations}
\item \emph{Input linear step:}  For each $j$,
\begin{subequations} \label{eq:inLinSimp}
\beqa
    1/\tau^r(t) &=& (1/n)\|\Abf\|^2_F\tau^s(t) \label{eq:murSimp}\\
    \rhat_j(t) &=& \xhat_j(t) + \tau^r(t)\sum_i a_{ij}\shat_i(t).
        \label{eq:rhatSimp}
\eeqa
\end{subequations}

\item \emph{Input nonlinear step:}
\begin{subequations} \label{eq:inNLSimp}
\beqa
    \lefteqn{\xhat_j(\tp1) = \gin(t,\rhat_j(t),q_j,\tau^r(t)) } \\
    \lefteqn{\tau^x(\tp1) = } \nonumber \\
    &  & \frac{\tau^r(t)}{n}\sum_{j=1}^n\frac{\partial}{\partial \rhat}
        \gin(t,\rhat_j(t),q_j,\tau^r(t)).  \label{eq:muxSimp}
\eeqa
\end{subequations}
Then increment $t=t+1$ and return to step 2 until a sufficient number
of iterations have been performed.
\end{enumerate}
\end{algorithm}


\section{Scalar Estimation Functions to Approximate Loopy BP} \label{sec:bpApprox}

As discussed in the previous section, through proper selection of
the scalar estimation functions $\gin(\cdot)$ and $\gout(\cdot)$,
GAMP can provide approximations of either max-sum or sum-product loopy BP.
With these selections, we can thus realize the two most useful
special cases of GAMP:
\begin{itemize}
\item \textbf{Max-sum GAMP}:  An approximation of max-sum loopy BP
for computations of the MAP estimates and computations of the marginal maxima of
the posterior; and
\item \textbf{Sum-product GAMP}:  An approximation of sum-product loopy BP
for computations of the MMSE estimates and the marginal posterior distributions.
\end{itemize}
Heuristic ``derivations" of the scalar estimation functions for both of these
algorithms
are sketched in Appendices \ref{sec:bpMap} and \ref{sec:bpMmse}
and summarized here.
The selection functions are also summarized in Table \ref{tbl:scalarEstim}.
We emphasize that the approximations are entirely heuristic --
we don't claim any formal properties of the approximation.
However, the analysis of the GAMP algorithm with these or other
functions will be more rigorous.

\renewcommand*\arraystretch{1.5}
\begin{table*}
\centering
\begin{tabular}{|l|l|l|l|l|}
\hline
\textbf{Method} &
    \multicolumn{2}{c|}{\textbf{Input scalar estimation functions}} &
    \multicolumn{2}{c|}{\textbf{Output scalar estimation functions}} \tabularnewline
    \cline{2-5}

 & $\gin(\rhat,q,\tau^r)$ & $\tau^r\gin'(\rhat,q,\tau^r)$ &
   $\gout(\phat,y,\tau^p)$ & $-\gout'(\phat,y,\tau^p)$ \tabularnewline \hline

\textbf{Max-sum GAMP} & $\argmax_x \Fin(z,\rhat,q,\tau^r)$ &
$\tau^r/(1-\tau^r\fin''(\xhat,q))$ & $(\zhat^0-\phat)/\tau^p$  &
    $\fout''(\zhat^0,y)/(\tau^p\fout''(\zhat^0,y)-1)$ \tabularnewline

for MAP estimation &  &  & $\zhat^0 := \argmax_z \Fout(z,\phat,y,\tau^p)$  &
    \tabularnewline \hline

\textbf{Sum-product GAMP} & $E(x|\rhat,q,\tau^r)$ & $\var(x|\rhat,q,\tau^r)$
  & $(\zhat^0-\phat)/\tau^p$  &  $(\tau^p(t)-\var(z|\phat,y))/(\tau^p(t))^2$
 \tabularnewline

for MMSE estimation  & \hspace{2mm}  $R=X+{\mathcal N}(0,\tau^r)$ &
    & $\zhat^0 := \Exp(z|\phat,y,\tau^p)$  &  \tabularnewline

  &  &  & \hspace{2mm} $Y\sim P_{Y|Z}$,
    $Z \sim {\mathcal N}(\phat,\tau^p)$  & \tabularnewline

\hline
\textbf{AWGN} & $\tau^{x0}(\rhat\!-\!q)/(\tau^r\!+\!\tau^{x0})\!+\!q$ &
  $\tau^{x0}\tau^r/(\tau^r+\tau^{x0})$
  & $(y-\phat)/(\tau^w+\tau^p)$  &  $1/(\tau^w+\tau^p)$
 \tabularnewline

 & \hspace{2mm} $X = {\mathcal N}(q,\tau^{x0})$ &
  & \hspace{2mm} $Y  = Z+ {\mathcal N}(0,\tau^w)$  &
 \tabularnewline \hline

\end{tabular}
\caption{Scalar input and output estimation functions
max-sum and sum-product GAMP.
For AWGN inputs and outputs, the scalar estimation functions for both max-sum
and sum-product algorithms are the same and have a particularly simple form. }
\label{tbl:scalarEstim}
\end{table*}

\renewcommand*\arraystretch{1.0}

\subsection{Max-Sum GAMP for MAP Estimation} \label{sec:bpApproxMap}

To describe the MAP estimator, observe that
for the linear mixing estimation problem in Section \ref{sec:intro},
the posterior density of the vector $\xbf$ given the system input $\qbf$
and output $\ybf$ is given by the conditional density function
\beq \label{eq:pxyq}
    p_{\xbf|\qbf,\ybf}(\xbf|\qbf,\ybf) := \frac{1}{Z(\qbf,\ybf)}
        \exp\left( F(\xbf,\Abf \xbf,\qbf,\ybf) \right),
\eeq
where
\beq \label{eq:Fmap}
    F(\xbf,\zbf,\qbf,\ybf) := \sum_{j=1}^n \fin(x_j,q_j)
        + \sum_{i=1}^m \fout(z_i,y_i),
\eeq
and
\begin{subequations} \label{eq:finout}
\beqa
    \fout(z,y) &:=& \log p_{Y|Z}(y|z) \label{eq:foutMap} \\
    \fin(x,q) &:=& \log p_{X|Q}(x|q). \label{eq:finMap}
\eeqa
\end{subequations}
The constant $Z(\qbf,\ybf)$ in \eqref{eq:pxyq} is a normalization constant.
Given this distribution, the \emph{maximum a posteriori} (MAP)
estimator is the maxima of \eqref{eq:pxyq} which is given by the optimization
\beq \label{eq:xhatMap}
    \xbfhat^{\rm map} := \argmax_{\xbf \in \R^n} F(\xbf,\zbf,\qbf,\ybf), \ \
    \zbfhat = \Abf\xbfhat.
\eeq
For each component $j$, we will also be interested in the \emph{marginal maxima}
of the posterior distribution
\beq \label{eq:margUtil}
    \Delta_j(x_j) := \max_{\xbf \backslash x_j}  F(\xbf,\zbf,\qbf,\ybf), \ \
    \zbfhat = \Abf\xbfhat,
\eeq
where the maximization is over vector $\xbf \in \R^n$ with a fixed value for the coefficient
$x_j$.  Note that the component $\xhat_j$ of the MAP estimate is given by
$\xhat_j = \argmax_{x_j} \Delta_j(x_j)$.
This, $\Delta_j(x_j)$ can be interpreted as the sensitivity
of the maxima to the value of the component $x_j$.

Note that one may also be interested in the optimization \eqref{eq:xhatMap}
where the objective is of the form \eqref{eq:Fmap}, but the functions
$\fin(\cdot)$ and $\fout(\cdot)$ are not derived from any density functions.
The max-sum GAMP method can be applied to these general optimization problems as well.

Now, an approximate implementation of max-sum BP for the MAP estimation
problem \eqref{eq:xhatMap} is described in Appendix \ref{sec:bpMap}.
It is suggested there that a possible input function to approximately
implement max-sum BP is given by
\beq \label{eq:ginMap}
    \gin(\rhat,q,\tau^r) := \argmax_x  \Fin(x,\rhat,q,\tau^r)
\eeq
where
\beq \label{eq:FinMap}
    \Fin(x,\rhat,q,\tau^r) := \fin(x,q) - \frac{1}{2\tau^r}(\rhat -x)^2.
\eeq
Here, and below, we have dropped the dependence on the iteration number $t$
when it is not needed.
The Appendix also shows that the function \eqref{eq:ginMap}
has a derivative satisfying
\beq \label{eq:ginMapDeriv}
    \tau^r \frac{\partial}{\partial \rhat} \gin(\rhat,q,\tau^r) =
        \frac{\tau^r }{1-\tau^r\fin''(\xhat,q)},
\eeq
where the second derivative $\fin''(x,q)$ is with respect to $x$
and $\xhat = \gin(r,q,\tau^r)$.
Also, the marginal maxima \eqref{eq:margUtil} is approximately given by
\beq \label{eq:DelFin}
    \Delta_j(x_j) \approx \Fin(x_j,\rhat_j,q_j,\tau_j^r)+ \mbox{const},
\eeq
where the constant term does not depend on $x_j$.

The initial condition for the GAMP algorithm should be taken as
\beq \label{eq:xinitMap}
    \xhat_j(0) = \argmax_{x_j} \fin(x_j,q_j), \ \ \
    \tau^x_j(0) = \frac{1}{\fin''(\xhat_j(0),q_j)}.
\eeq
This initial conditions corresponds to the output of \eqref{eq:inNL} with
$t=0$, the functions in \eqref{eq:ginMap}
and \eqref{eq:ginMapDeriv} and $\tau^r(-1) \arr \infty$.

Observe that when $\fin$ is given by \eqref{eq:finMap}, $\gin(\rhat,q,\tau^r)$
is precisely the scalar MAP estimate of a random variable $X$
given $Q=q$ and $\Rhat=\rhat$ for the random variables
\beq \label{eq:RhatV}
    \Rhat = X + V, \ \ \ V \sim {\cal N}(0,\tau^r),
\eeq
where $X \sim p_{X|Q}(x|q)$, $Q \sim p_Q(q)$ and
with $V$ independent of $X$ and $Q$.  With this definition,
$\Rhat$ can be interpreted as a Gaussian noise-corrupted version of $X$
with noise level $\tau^r$.

Appendix \ref{sec:bpMap} shows that the
output function for the approximation of max-sum BP is given by
\beq \label{eq:goutMap}
    \gout(\phat,y,\tau^p) := \frac{1}{\tau^p}(\zhat^0 - \phat),
\eeq
where
\beq \label{eq:zhat0Map}
    \zhat^0 := \argmax_z \Fout(z,\phat,y,\tau^p),
\eeq
and
\beq \label{eq:FoutMap}
    \Fout(z,\phat,y,\tau^p) := \fout(z,y) - \frac{1}{2\tau^p}(z-\phat)^2.
\eeq
The negative derivative of this function is given by
\beq \label{eq:goutMapDeriv}
    -\frac{\partial}{\partial \phat} \gout(\phat,y,\tau^p) =
        \frac{-\fout''(\zhat^0,y)}{1-\tau^p\fout''(\zhat^0,y)},
\eeq
where the second derivative $\fout''(z,y)$ is with respect to $z$.

Now, when $\fout(z,y)$ is given by \eqref{eq:foutMap},
$\Fout(z,\phat,y,\tau^p)$ in \eqref{eq:FoutMap}
can be interpreted as the log posterior of a random variable $Z$
given $Y = y$ and
\beq \label{eq:Zchan}
    Z \sim {\cal N}(\phat,\tau^p), \qquad Y \sim p_{Y|Z}(y|z).
\eeq
In particular, $\zhat^0$ in \eqref{eq:zhat0Map} is the MAP estimate
of $Z$.

We see that the max-sum GAMP algorithm reduces the vector MAP estimation
problem to a sequence of scalar MAP estimations problems at the inputs
and outputs.
The scalar parameters $\tau^r(t)$ and $\tau^p(t)$
represent effective Gaussian noise levels in these problems.
The equations for the scalar estimation functions are summarized in
Table \ref{tbl:scalarEstim}.

\subsection{Sum-Product GAMP for MMSE Estimation} \label{sec:bpApproxMmse}

The minimum mean squared error (MMSE) estimate
is the conditional expectation
\beq \label{eq:xhatMmse}
    \xbfhat^{\rm mmse} := \Exp\left[\xbf ~|~\ybf,\qbf\right],
\eeq
relative to the conditional density \eqref{eq:pxyq}.
We are also interested in the log posterior marginals
\beq \label{eq:logpxjPost}
    \Delta_j(x_j) := \log p(x_j|\qbf,\ybf).
\eeq

The selection of the functions $\gin(\cdot)$ and $\gout(\cdot)$
to approximately implement sum-product loopy BP to compute the
MMSE estimates and posterior marginals
is described in Appendix \ref{sec:bpMmse}.
Heuristic arguments in that section, show that
the input function to implement BP-based MMSE estimation is given by
\beq \label{eq:ginMmse}
    \gin(\rhat,q,\tau^r) := \Exp[X~|~\Rhat=\rhat, ~Q=q],
\eeq
where the expectation is over the variables in \eqref{eq:RhatV}.
Also, the derivative is given by the variance,
\beq \label{eq:ginMmseDeriv}
   \tau^r \frac{\partial}{\partial \rhat} \gin(\rhat,q,\tau^r)
        := \var[X~|~\Rhat=\rhat, ~Q=q].
\eeq
In addition, the log posterior marginal \eqref{eq:logpxjPost}
is approximately given by \eqref{eq:DelFin}.  From the definition of $\Fin(\cdot)$
in \eqref{eq:FinMap} we see that the posterior marginal is approximately given by
\beq \label{eq:pxjPost}
    p(x_j|\qbf,\ybf) \approx \frac{1}{Z}
        p_{X|Q}(x_j|q_j)\exp\Bigl[ -\frac{1}{2\tau_r}(\rhat_j-x_j)^2\Bigr],
\eeq
where $Z$ is a normalization constant.

The initial condition for the GAMP algorithm for MMSE estimation
should be taken as
\beq \label{eq:xinitMmse}
    \xhat_j(0) = \Exp(X~|~Q=q_j) \ \
    \tau^x_j(0) = \var(X~|~Q=q_j),
\eeq
where the expectations are over the distribution $p_{X|Q}(x_j|q_j)$.
Thus, the algorithm is initialized to the the prior mean and variance
on $x_j$ based on the parameter $q_j$ but no observations.
Equivalently, the initial conditions \eqref{eq:xinitMmse} corresponds to the output of \eqref{eq:inNL} with
$t=0$, the functions in \eqref{eq:ginMmse}
and \eqref{eq:ginMmseDeriv} and $\tau^r(-1) \arr \infty$.

To describe the output function $\gout(\phat,y,\tau^p)$,
consider a random variable $z$ with conditional probability density
\beq \label{eq:pFout}
    p(z|\phat,y,\tau^p) \propto \exp \Fout(z,\phat,y,\tau^p),
\eeq
where $\Fout(z,\phat,y,\tau^p)$ is given in \eqref{eq:FoutMap}.
The distribution \eqref{eq:pFout} is the posterior density function
of the random variable $Z$ with observation $Y$ in \eqref{eq:Zchan}.
Appendix \ref{sec:bpMmse} shows that the output function
 $\gout(\phat,y,\tau^p)$ to implement approximate BP for the MMSE problem
is given by
\beq \label{eq:goutMmse}
    \gout(\phat,y,\tau^p) := \frac{1}{\tau^p}(\zhat^0 - \phat), \ \
        \zhat^0 := \Exp(z|\phat,y,\tau^p),
\eeq
where the expectation is over the density function \eqref{eq:pFout}.
Also, the negative derivative of $\gout(\cdot)$ is given by
\beq \label{eq:goutMmseDeriv}
    -\frac{\partial}{\partial \phat} \gout(\phat,y,\tau^p) =
        \frac{1}{\tau^p}\left(1 -
         \frac{\var\left(z|\phat,y,\tau^p\right)}{\tau^p}\right).
\eeq

Appendix \ref{sec:bpMmse}
also provides an alternative definition for $\gout(\cdot)$:
The function $\gout(\cdot)$ in \eqref{eq:goutMmse} is given by
\beq \label{eq:goutMmseScore}
    \gout(\phat,y,\tau^p) := \frac{\partial}{\partial \phat} \log
    p(y|\phat,\tau^p),
\eeq
where $p(y|\phat,\tau^p)$ is the density is from the channel \eqref{eq:Zchan}.
As a result, its negative derivative is
\beq \label{eq:goutMmseScoreDeriv}
    -\frac{\partial}{\partial \phat} \gout(\phat,y,\tau^p) =
    -\frac{\partial^2}{\partial \phat^2} \log p(y|\phat,\tau^p).
\eeq
Hence $\gout(\phat,y,\tau^p)$ has the interpretation of a \emph{score
function} of the parameter $\phat$ in the distribution of the random
variable $Y$ in \eqref{eq:Zchan}.

Thus, similar to the MAP estimation problem,
the sum-product GAMP algorithm reduces the vector MMSE estimation
problem to the evaluation of sequence of
scalar estimation problems from Gaussian noise.
Scalar MMSE estimation is performed at the input nodes,
and a score function of an ML estimation problem is performed at the output nodes.

\subsection{AWGN Output Channels} \label{sec:AWGN}
In the special case of an AWGN output channel, we will see that
that the output functions for max-sum and sum-product GAMP
are identical and reduce to the update in
the AMP algorithm of Bayati and Montanari in \cite{BayatiM:11}.
Suppose that the output channel is described by the AWGN distribution
\eqref{eq:pyzAwgn} for some output variance $\tau^w > 0$.
Then, it can be checked that
the distribution $p(z|\phat,y,\tau^p)$ in \eqref{eq:pFout} is the Gaussian
\beq \label{eq:pFoutAwgn}
    p(z|\phat,y,\tau^p) \sim {\cal N}(\zhat^0, \tau^z),
\eeq
where
\begin{subequations} \label{eq:zmuAwgn}
\beqa
    \zhat^0 &:=& \phat + \frac{\tau^p}{\tau^w + \tau^p}(y-\phat),
        \label{eq:zhat0Awgn} \\
    \tau^z &:=& \frac{\tau^w\tau^p}{\tau^w + \tau^p} \label{eq:muzAwgn}
\eeqa
\end{subequations}
It can be verified that the conditional mean $\zhat^0$ in \eqref{eq:zhat0Awgn}
agrees with both $\zhat^0$ in
\eqref{eq:zhat0Map} for the MAP estimator
and $\zhat^0$ in \eqref{eq:goutMmse} for the MMSE estimator.
Therefore, the output function $\gout$ for both the MAP estimate in
\eqref{eq:goutMap} and MMSE estimate in \eqref{eq:goutMmse} is given by
\beq \label{eq:goutAwgn}
    \gout(\phat,y,\tau^p) := \frac{y-\phat}{\tau^w + \tau^p}.
\eeq
The negative derivative of the function is given by
\beq \label{eq:goutAwgnDeriv}
    -\frac{\partial}{\partial \phat} \gout(\phat,y,\tau^p) =
     \frac{1}{\tau^p+\tau^w}.
\eeq
Therefore, for both max-sum and sum-product GAMP
$\gout(\phat,y,\tau^p)$ and its derivative are given by
\eqref{eq:goutAwgn} and \eqref{eq:goutAwgnDeriv}.
When, we apply these equations into the input updates \eqref{eq:inNL},
we precisely obtain the original AMP algorithm (with some scalings)
given in \cite{BayatiM:11}.

\subsection{AWGN Input Channels} \label{sec:AwgnIn}
Now suppose that the input density function $p_{X|Q}(x|q)$ is a Gaussian
density
\[
    p_{X|Q}(x|q) = {\mathcal N}(q,\tau^{x0}),
\]
for some variance $\tau^{x0} > 0$.  Then, it is easily checked that the
input estimate function $\gin(\cdot)$ and its derivative
are identical for both max-sum and sum-product GAMP and given by
\begin{subequations} \label{eq:ginAwgn}
\beqa
    \gin(\rhat,q,\tau^r) &:=& \frac{\tau^{x0}}{\tau^{x0}+ \tau^r}(\rhat-q) + q \\
    \tau^r\gin'(\rhat,q,\tau^r) &:=& \frac{\tau^{x0}\tau^r}{\tau^{x0}+ \tau^r}.
\eeqa
\end{subequations}
The functions are shown in Table \ref{tbl:scalarEstim}.

\section{Asymptotic Analysis} \label{sec:asymAnal}

We now present our main theoretical result, which is the SE analysis of
GAMP for large, Gaussian i.i.d.\ matrices $\Abf$.  

\subsection{Empirical Convergence of Random Variables} \label{sec:conv}

The analysis is a relatively minor modification of the results in
Bayati and Montanari's paper \cite{BayatiM:11}.
The work \cite{BayatiM:11} employs certain deterministic
models on the vectors and then proves convergence properties
of related empirical distributions. To apply the same analysis here, we need
to review some of their definitions.  We say a function $\phi:\R^r \arr \R^s$
is \emph{pseudo-Lipschitz} of order $k>1$, if there exists an $L > 0$
such for any $\xbf$, $\ybf \in \R^r$,
\[
    \|\phi(\xbf) - \phi(\ybf)\|
        \leq L(1+\|\xbf\|^{k-1}+\|\ybf\|^{k-1})\|\xbf-\ybf\|.
\]

Now suppose that for each $n=1,2,\ldots$,
$\vbf^{(n)}$ is a block vector with components
$\vbf^{(n)}_i \in \R^s$, $i=1,2,\ldots,\ell(n)$ for some $\ell(n)$.
So, the total dimension of $\vbf^{(n)}$ is $s\ell(n)$.
We say that the components of the vectors $\vbf^{(n)}$
\emph{empirically converges with bounded moments of order $k$} as $n\arr\infty$
to a random vector $\Vbf$ on $\R^s$ if:  For all
pseudo-Lipschitz continuous functions, $\phi$, of order $k$,
\[
    \lim_{n \arr\infty} \frac{1}{\ell(n)} \sum_{i=1}^{\ell(n)} \phi(\vbf^{(n)}_i)
    = \Exp(\phi(\Vbf)) < \infty.
\]
When the nature of convergence
is clear, we may write (with some abuse of notation)
\[
      \lim_{n \arr \infty} \vbf^{(n)}_i \stackrel{d}{=} \Vbf.
\]

\subsection{Assumptions} \label{sec:assumptions}

With these definitions, we can now formally state the modeling assumptions,
which follow along the lines of the
asymptotic model considered by Bayati and Montanari
in \cite{BayatiM:11}.
Specifically, we consider a sequence of random realizations of the
estimation problem in Section \ref{sec:intro}
indexed by the input signal dimension $n$.  For each $n$, we
assume the output dimension $m=m(n)$
is deterministic and scales linearly with the input dimension in that
\beq \label{eq:betaDef}
    \lim_{n \arr \infty} n/m(n) = \beta,
\eeq
for some $\beta > 0$ called the \emph{measurement ratio}.
We also assume that the transform matrix $\Abf \in \R^{m \x n}$
has i.i.d.\ Gaussian
components $a_{ij} \sim {\cal N}(0,1/m)$ and $\zbf = \Abf\xbf$.

We assume that for some order $k \geq 2$, the components
of initial condition $\xbfhat(0)$, $\taubf^x(0)$
and input vectors $\xbf$ and $\qbf$
empirically converge with bounded moments of order $2k-2$ as
\beq \label{eq:inputLim}
   \lim_{n \arr \infty} (\xhat_j(0),\tau^x_j(0),x_j,q_j) \stackrel{d}{=}
   (\Xhat_0,\taubar^x(0),X,Q),
\eeq
for some random variable triple $(\Xhat_0,X,Q)$ with
joint density $p_{\Xhat_0,X,Q}(\xhat_0,x,q)$ and constant
$\taubar^x(0)$.

To model the dependence of the system output vector
$\ybf$ on the transform output $\zbf$, we assume that, for each $n$,
there is a deterministic vector $\wbf \in W^m$ for
some set $W$, which we can think of as a noise vector.
Then, for every $i=1,\ldots,m$, we assume that
\beq \label{eq:hout}
    y_i = h(z_i,w_i)
\eeq
where $h$ is some function $h:\R \x W \arr Y$ and $Y$ is the output set.
Finally, we assume that the components of $\wbf$ empirically
converge with bounded moments of order $2k-2$
to some random variable $W$ with distribution $p_W(w)$.
We will write $p_{Y|Z}(y|z)$ for the conditional distribution
of $Y$ given $Z$ in this model.

We also need certain continuity assumptions.  Specifically, we assume that
the partial derivatives of the functions $\gin(t,\rhat,q,\tau^r)$
and $\gout(t,\phat,h(z,w),\tau^p)$ with respect
to $\rhat$, $\phat$ and $z$ exist almost everywhere
and are pseudo-Lipschitz continuous of order $k$.
This assumption implies that the functions $\gin(t,\rhat,q,\tau^r)$
and $\gout(t,\phat,y,\tau^p)$ are Lipschitz continuous in
$\rhat$ and $\phat$ respectively.

\subsection{State Evolution Equations} \label{sec:stateEvo}
Similar to \cite{BayatiM:11}, the key result here is that the
behavior of the GAMP algorithm is described by a set of state evolution (SE)
equations.  For the GAMP algorithm, the SE equations are easiest to describe
algorithmically as shown in 
in Algorithm \ref{algo:SE}.  To describe the equations, we introduce
two random vectors -- one corresponding to the input channel,
and the other corresponding to the output channel.
At the input channel, given $\alpha^r$, $\xi^r \in \R$, 
define the random variable triple
\beq \label{eq:thetaR}
    \theta^r(\xi^r,\alpha^r) := (X,Q,\Rhat),
\eeq
where the distribution $(X,Q)$ are derived from the 
density in \eqref{eq:inputLim} and $\Rhat$ is given
\beq \label{eq:RhatValpha}
    \Rhat = \alpha^r X + V, \ \ V \sim {\cal N}(0,\xi^r),
\eeq
with $V$ independent of $X$ and $Q$.
At the output channel,
given a covariance matrix $\Kbf^p \in \R^{2\x2}$, $\Kbf^p > 0$,
 define the four-dimensional random vector
\beq \label{eq:thetaP}
    \theta^p(\Kbf^p) := (Z,\Phat,W,Y), \ \ Y=h(Z,W),
\eeq
such that $W$ and $(Z,\Phat)$ are independent with distributions
\beq \label{eq:zphatCov}
    (Z,\Phat) \sim {\cal N}(0,\Kbf^p), \ \ W \sim p_W(w).
\eeq
Since the computations in \eqref{eq:outSE} and \eqref{eq:inSE}
are expectations over scalar random variables, they can be evaluated
numerically given functions $\gin(\cdot)$ and $\gout(\cdot)$.

\begin{algorithm}
\caption{GAMP State Evolution }
\label{algo:SE}
Given scalar estimation functions $\gin(\cdot)$ and $\gout(\cdot)$,
the density function $p_{\Xhat_0,X,Q}$ and initial value $\taubar^x(0)$ in \eqref{eq:inputLim},
the density function  $p_{Y|Z}$ at the output
and the measurement ratio $\beta = m/n$, compute the state evolution
parameters as follows:
\begin{enumerate}
\item \emph{Initialization:}  Set $t=0$, let $\taubar^x(0)$
be the initial value in \eqref{eq:inputLim} and set
\beq \label{eq:SEinit}
    \Kbf^x(0) = \cov(X,\Xhat(0)),
\eeq
meaning the covariance matrix of the random variables $(X,\Xhat(0))$
in the limit \eqref{eq:inputLim}.
\item \emph{Output node update:}  Compute
\begin{subequations} \label{eq:outSE}
\beqa
    \taubar^p(t) &=& \beta\taubar^x(t), \ \ \ \Kbf^p(t) = \beta \Kbf^x(t)
        \label{eq:outSEKp} \\
    \taubar^r(t) &=& -\Exp^{-1}\left[ \frac{\partial}{\partial \phat}
        \gout(t,\Phat,Y,\taubar^p(t)) \right] \label{eq:outSEmur} \\
    \xi^r(t) &=& (\taubar^r(t))^2\Exp\left[ \gout^2(t,\Phat,Y,\taubar^p(t)) \right]
    \label{eq:taurSE}
\eeqa
where the expectations are over the random variable triples
$\theta^p(\Kbf^p(t)) = (Z,\Phat,W,Y)$ in \eqref{eq:thetaP}.
Also, let
\beqa
    \lefteqn{ \alpha^r(t) = \taubar^r(t)} \nonumber \\
    &\x&\Exp\left[ \left. \frac{\partial}{\partial z}
        \gout(t,\Phat,h(z,W),\taubar^p(t)) \right|_{z=Z} \right]. \hspace{0.25in}
        \label{eq:alphaSE}
\eeqa
\end{subequations}
\item \emph{Input node update:}  Compute
\begin{subequations} \label{eq:inSE}
\beqa
    \taubar^x(\tp1) &=& \taubar^r(t)\Exp\left[ \frac{\partial}{\partial \rhat}
        \gin(t,Q,\Rhat,\taubar^r(t)) \right] \label{eq:inSEmux} \\
    \Kbf^x(\tp1) &=& \cov(X,\Xhat(\tp1)) \hspace{0.25in} \label{eq:inSEKx}
\eeqa
\end{subequations}
where the expectation is over the random variable triple
$\theta^r(\xi^r(t),\alpha^r(t)) = (X,Q,\Rhat)$ in \eqref{eq:thetaR},
and
\[
    \Xhat(\tp1) = \gin(t,Q,\Rhat,\taubar^r(t)).
\]
Increment $t=t+1$ and return to step 2.
\end{enumerate}
\end{algorithm}

\subsection{Main Result} \label{sec:mainResult}
To simplify the analysis, we consider the GAMP method with scalar variances,
Algorithm \ref{algo:GAMP-simp}.  Our simulations indicate no difference between
Algorithms \ref{algo:GAMP} and \ref{algo:GAMP-simp} at moderate block lengths.
We also assume the following minor modifications
\begin{itemize}
\item The variance $\tau^r(t)$ in \eqref{eq:inLinSimp}
is replaced by its deterministic limit $\taubar^r(t)$;
\item The variance $\tau^p(t)$ in \eqref{eq:outLinSimp}
is replaced by its deterministic limit $\taubar^p(t)$; and
\item The norm $\|\Abf\|^2_F$ is replaced by its expectation $\Exp\|\Abf\|^2_F=n$.
\end{itemize}
These simplifications are likely not significant, since we will show that,
under these simplifications,
for all $t$, $\tau^r(t) \arr \taubar^r(t)$ and $\tau^p(t) \arr \taubar^p(t)$.
Also, since $\Abf$ has i.i.d.\ components with variance $1/m$, $(1/n)\Exp\|\Abf\|^2 \arr 1$.
Moreover, it is possible that one could formally justify the simplification with
the arguments in \cite{KamRanFU:12-arXiv}.

\begin{claim}  \label{thm:gampSE}  Consider the GAMP
with scalar variances, Algorithm \ref{algo:GAMP-simp}, under the assumptions
in Section \ref{sec:assumptions} and with the above modifications.
Then, for any fixed iteration number $t$:
\begin{itemize}
\item[(a)]  Almost surely, we have the limits
\beq \label{eq:muLim}
    \lim_{n \arr \infty} \tau^r(t) = \taubar^r(t), \ \
    \lim_{n \arr \infty} \tau^p(t) = \taubar^p(t).
\eeq
\item[(b)]  The components of the vectors $\xbf$, $\qbf$, $\rbfhat$ and $\xbfhat$
empirically converge with bounded moments of order $k$ as
\beq \label{eq:thetarlim}
    \lim_{n \arr \infty} (x_j,q_j,\rhat_j(t))
    \stackrel{d}{=} \theta^r(\xi^r(t),\alpha^r(t)),
\eeq
where $\theta^r(\xi^r(t),\alpha^r(t)) = (X,Q,\Rhat)$
is the random variable triple
in \eqref{eq:thetaR} and $\xi^r(t)$ is from the SE equations.

\item[(c)] The components of the vectors $\zbf$, $\pbfhat$, $\wbf$ and
$\ybf$ empirically converge with bounded moments of order $k$ as
\beq \label{eq:thetaplim}
    \lim_{n \arr \infty} (z_i,\phat_i(t),w_i,y_i)
    \stackrel{d}{=} \theta^p(\Kbf^p(t)),
\eeq
where $\theta^p(\Kbf^p(t)) = (Z,\Phat,W,Y)$ is the random vector
in \eqref{eq:thetaP} and $\Kbf^p(t)$ is given by the SE equations.
\end{itemize}
\end{claim}

A proof of the claim is sketched in Appendix~\ref{sec:vecAMP}. 
We use the term ``claim" here since the proof relies on an extension of a
general result from \cite{BayatiM:11} to vector-valued quantities.  
Due to space considerations, we only sketch a proof of the extension.
The term ``claim", as opposed to ``theorem," 
is used to emphasize that the details have been omitted.

\begin{figure}
\begin{center}
  \includegraphics[width=3.2in,height=1.2in]{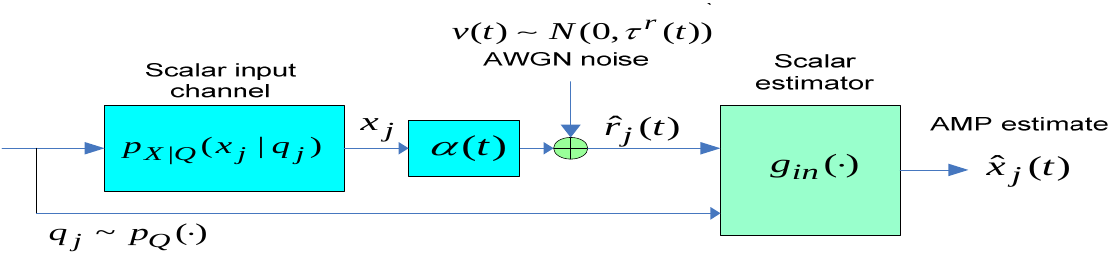}
\end{center}
\caption{Scalar equivalent model.  The joint distribution of any one
component $x_j$ and its estimate $\xhat_j(t)$ from the $t$th
iteration of the AMP algorithm is identical asymptotically
to an equivalent estimator with Gaussian noise.}
\label{fig:scaModAMP}
\end{figure}

A useful interpretation of Claim~\ref{thm:gampSE} is that it provides a
\emph{scalar equivalent model} for the behavior of the GAMP estimates.
The equivalent scalar model is illustrated in Fig.\ \ref{fig:scaModAMP}.
To understand this diagram, observe that part (b) of Claim~\ref{thm:gampSE}
shows that the joint empirical distribution of the components $(x_j,q_j,\rhat_j(t))$
converges to $\theta^r(\xi^r(t),\alpha^r(t)) = (X,Q,\Rhat)$
in \eqref{eq:thetaR}.  This distribution is identical
to $\rhat_j(t)$ being a scaled and noise-corrupted version of $x_j$.
Then, the estimate $\xhat_j(t) = \gin(t,\rhat_j(t),q_j,\tau^r_j(t))$
is a nonlinear scalar function of $\rhat_j(t)$ and $q_j$.
A similar interpretation can be drawn at the output nodes.

\section{Special Cases} \label{sec:special}
We now show that several previous results can be recovered from
the general SE equations in Section \ref{sec:stateEvo}
as special cases.

\subsection{AWGN Output Channel} \label{sec:awgnSE}

First consider the case where the  output function
\eqref{eq:hout} is given by additive noise model
\beq \label{eq:houtAdd}
    y_i = h(z_i,w_i) = z_i + w_i.
\eeq
Assume the components $w_i$ of the output noise vector
empirically converge to a random
variable $W$ with zero mean and variance $\tau^w > 0$.
The output noise $W$ need not be Gaussian. But,
suppose the GAMP algorithm uses an output
function $\gout(\cdot)$ in \eqref{eq:goutAwgn}
corresponding to an AWGN output for some
\emph{postulated} output variance $\tau^w_{\rm post}$ that may differ
from $\tau^w$.  This is the scenario analyzed in Bayati and Montanari's
paper \cite{BayatiM:11}.

Substituting \eqref{eq:goutAwgnDeriv} with the postulated noise
variance $\tau^w_{\rm post}$ into \eqref{eq:outSEmur} we get
\beq \label{eq:murAwgn}
    \taubar^r(t) = \tau^w_{\rm post} + \taubar^p(t) =
    \tau^w_{\rm post} + \beta\taubar^x(t).
\eeq
Also, using \eqref{eq:goutAwgn} and \eqref{eq:houtAdd},
we see that
\[
\frac{\partial}{\partial z} \gout(\phat,h(z,y),\taubar^p(t))
    = \frac{\partial}{\partial z} \frac{z+w-\phat}{\tau^w_{\rm post} + \taubar^p(t)}
    = \frac{1}{\taubar^r(t)}.
\]
Therefore, \eqref{eq:alphaSE} implies that  $\alpha^r(t) = 1$.

Now define
\begin{subequations} \label{eq:tauxp}
\beqa
    \xi^x(t) &:=& \left[1 \ -1\right]\Kbf^x(t)
    \left[\begin{array}{c}1 \\ -1 \end{array}\right] \\
    \xi^p(t) &:=& \beta \xi^x(t) = \left[1 \ -1\right]\Kbf^p(t)
    \left[\begin{array}{c}1 \\ -1 \end{array}\right].
\eeqa
\end{subequations}
With these definitions,
note that if $(Z,\Phat) \sim {\cal N}(0,\Kbf^p(t))$, then
\beq \label{eq:taupAwgn}
    \xi^p(t) = \Exp(Z-\Phat)^2.
\eeq
Also,  with $\Kbf^x(t)$ defined in \eqref{eq:inSEKx},
\beqa
    \lefteqn{ \xi^x(\tp1) = \Exp(X-\Xhat(t))^2 } \nonumber \\
    &=&
    \Exp\left[X-\gin(t,\Rhat,Q,\taubar^r(t))\right]^2.\label{eq:tauxAwgn}
\eeqa
Therefore,
\beqa
   \lefteqn{ \xi^r(\tp1) \stackrel{(a)}{=} (\taubar^r(\tp1))^2
    \frac{\Exp(Y-\Phat)^2}{(\tau^w_{\rm post} + \taubar^p(\tp1))^2} }\nonumber \\
    &\stackrel{(b)}{=}& \Exp(Y-\Phat)^2 \nonumber \\
    &\stackrel{(c)}{=}& \Exp(W+Z-\Phat)^2 = \tau^w + \Exp(Z-P)^2 \nonumber \\
    &\stackrel{(d)}{=}& \tau^w + \xi_p(\tp1) \nonumber \\
    &\stackrel{(e)}{=}& \tau^w +
        \beta\Exp\left[X-\gin(t,\Rhat,Q,\taubar^r(t))\right]^2,
        \label{eq:taurAwgn}
\eeqa
where (a) follows from substituting \eqref{eq:goutAwgn}
into \eqref{eq:taurSE};
(b) follows from \eqref{eq:murAwgn};
(c) follows from the output channel model assumption \eqref{eq:houtAdd};
(d) follows from \eqref{eq:taupAwgn} and
(e) follows from \eqref{eq:tauxp} and \eqref{eq:tauxAwgn}.
The expectation in \eqref{eq:taurAwgn} is over
$\theta^r(\xi^r(t),\alpha^r(t)) = (X,Q,\Rhat)$
in \eqref{eq:thetaR}.
Since $\alpha^r(t) = 1$, the expectation is over random variables
$(X,Q,\Rhat)$ where $X \sim p_{X|Q}(x|q)$, $Q \sim p_Q(q)$ and
\beq \label{eq:RhatVtau}
    \Rhat = X + V, \ \ \ V \sim {\cal N}(0,\xi^r(t)),
\eeq
and $V$ is independent of $X$ and $Q$.
With this expectation, \eqref{eq:taurAwgn} is precisely the SE equation
in \cite{BayatiM:11} for the
AMP algorithm with a general input function $\gin(\cdot)$.

\subsection{Sum-Product GAMP with AWGN Output Channels}

The SE equation \eqref{eq:taurAwgn} applies to a general
input function $\gin(\cdot)$.
Now suppose that the estimator uses an input function
$\gin(\cdot)$ based on the sum-product estimator \eqref{eq:ginMmse}.
To account for possible mismatches with the estimator, suppose that the
sum-product GAMP algorithm is run with some
postulated distribution $\ppost_{X|Q}(\cdot)$ that may differ from
the true distribution $p_{X|Q}(\cdot)$.
Let $\xhat^{\rm post}_{\rm mmse}(\rhat,q,\tau^r)$ be the corresponding
scalar MMSE estimator
\[
    \xhat^{\rm post}_{\rm mmse}(\rhat,q,\tau^r) :=
        \Exp\left( X|\Rhat=\rhat,Q=q\right),
\]
where the expectation is over the random variable
\beq \label{eq:RhatVmu}
    \Rhat = X + V, \ \ \ V \sim {\cal N}(0,\tau^r),
\eeq
where $X$ given $Q$
follows the postulated distribution $\ppost_{X|Q}(x|q)$ and
$V$ is independent of $X$ and $Q$.
Let ${\cal E}^{\rm post}_{\rm mmse}(\rhat,q,\tau^r)$ denote the corresponding
postulated variance.
Then, \eqref{eq:ginMmse} and \eqref{eq:ginMmseDeriv} can be re-written
\begin{subequations}
\beq \label{eq:ginMmseAwgn}
    \gin(\rhat,q,\tau^r) = \xhat^{\rm post}_{\rm mmse}(\rhat,q,\tau^r),
\eeq
and
\beq \label{eq:ginDerivMmseAwgn}
    \tau^r \frac{\partial}{\partial \rhat} \gin(\rhat,q,\tau^r)
        = {\cal E}^{\rm post}_{\rm mmse}(\rhat,q,\tau^r).
\eeq
\end{subequations}
With this notation,
\beqa
    \lefteqn{ \taubar^r(\tp1) \stackrel{(a)}{=}
        \tau^w_{\rm post} + \beta\taubar^x(\tp1) } \nonumber \\
    &\stackrel{(b)}{=}& \tau^w_{\rm post} + \beta
        \Exp\left[ {\cal E}^{\rm post}_{\rm mmse}(\rhat,q,\taubar^r(t)) \right]
        \label{eq:mubarAwgnMmse}
\eeqa
where (a) follows from \eqref{eq:murAwgn} and
(b) follows from substituting \eqref{eq:ginDerivMmseAwgn}
into \eqref{eq:inSEmux}.
Also, substituting \eqref{eq:ginMmseAwgn} into
\eqref{eq:taurAwgn}, we obtain
\beq \label{eq:taurAwgnMmse}
    \xi^r(\tp1) = \tau^w
    + \beta\Exp\left[X-\xhat^{\rm post}_{\rm mmse}(\Rhat,Q,\taubar^r(t))\right]^2.
\eeq

The fixed point of the updates \eqref{eq:mubarAwgnMmse}
and \eqref{eq:taurAwgnMmse} are precisely the equations given by Guo and Verd{\'u}
in their replica analysis of MMSE estimation with AWGN outputs
and non-Gaussian priors \cite{GuoV:05}.
Their work uses the replica method from statistical physics
to argue the following:  Let $\xbfhat^{\rm post}$ be the exact
MMSE estimate of the vector $\xbf$ based on the postulated distributions
$\ppost_{X|Q}$ and postulated noise variance $\tau^w_{\rm post}$.
By ``exact," we mean the actual MMSE estimate for that postulated
prior -- not the GAMP or any other approximation.
Then, in the limit of large i.i.d.\ random matrices $\Abf$,
it is claimed that the joint distribution of
$(x_j,q_j)$ for some component $j$ and the corresponding
exact MMSE estimate $\xhat^{\rm post}_j$, follows the same scalar
model as Fig.\ \ref{fig:scaModAMP} for the GAMP algorithm.
Moreover, the effective noise variance of the exact MMSE estimator
is a fixed point of the updates \eqref{eq:mubarAwgnMmse}
and \eqref{eq:taurAwgnMmse}.
This result of Guo and Ver{\'u}
generalized several earlier results including analyses of linear estimators
in \cite{VerduS:99} and \cite{TseH:99} based on
random matrix theory and BPSK estimation in \cite{Tanaka:02} based on
the replica method.
As discussed in \cite{BayatiM:11},
the SE analysis of the AMP algorithm
provides some rigorous basis for the replica analysis, along with
an algorithm to achieve the performance.

Also, when the postulated distribution matches the true distribution,
the general equations \eqref{eq:mubarAwgnMmse}
and \eqref{eq:taurAwgnMmse}
reduce to the SE equations for Gaussian approximated BP on large sparse matrices
derived originally by Boutros and Caire \cite{BoutrosC:02} along
with other works including \cite{GuoW:06} and \cite{MontanariT:06}.
In this regard, the analysis here can be seen as an extension of that
work to handle dense matrices and the case of possibly mismatched distributions.

\subsection{Max-Sum GAMP with AWGN Output Channels}

Now suppose that the estimator uses an input function
$\gin(\cdot)$ based on the MAP estimator \eqref{eq:ginMap} again
with a postulated distribution $\ppost_{X|Q}(\cdot)$ that may differ from
the true distribution $p_{X|Q}(\cdot)$.
Let $\xhat^{\rm post}_{\rm map}(\rhat,q,\tau^r)$ be the corresponding
scalar MAP estimator
\beq \label{eq:xhatMapPost}
    \xhat^{\rm post}_{\rm map}(\rhat,q,\tau^r) :=
        \argmax_{x \in \R} \left[ \fin(x,q) - \frac{1}{2\tau^r}
        (\rhat -x)^2 \right],
\eeq
where
\[
    \fin(x,q) = \log \ppost_{X|Q}(x|q).
\]
With this definition, $\xhat^{\rm post}_{\rm map}$ is
the scalar MAP estimate of the random variable $X$ from the
observations $\Rhat=\rhat$ and $Q=q$ in the model \eqref{eq:RhatVmu}.
Also, following \eqref{eq:ginMapDeriv} define
\[
    {\cal E}^{\rm post}_{\rm map}(\rhat,q,\tau^r) :=
    \frac{\tau^r}{1-\fin''(\xhat,q)\tau^r},
\]
where $\xhat = \xhat^{\rm post}_{\rm map}$.
Then, \eqref{eq:ginMap} and \eqref{eq:ginMapDeriv} can be re-written as
\begin{subequations}
\beq \label{eq:ginMapAwgn}
    \gin(\rhat,q,\tau^r) = \xhat^{\rm post}_{\rm map}(\rhat,q,\tau^r),
\eeq
and
\beq \label{eq:ginDerivMapAwgn}
    \tau^r \frac{\partial}{\partial \rhat} \gin(\rhat,q,\tau^r)
        = {\cal E}^{\rm post}_{\rm map}(\rhat,q,\tau^r).
\eeq
\end{subequations}
Following similar computations to the previous subsection,
we obtain the SE equations
\begin{subequations}
\beqa
    \taubar^r(\tp1) &=& \tau^w_{\rm post} + \beta
        \Exp\left[ {\cal E}^{\rm post}_{\rm map}(\Rhat,Q,\taubar^r(t)) \right] \\
    \xi^r(\tp1) &=& \tau^w
    + \beta\Exp\left[X-\xhat^{\rm post}_{\rm map}(\Rhat,Q,\taubar^r(t))\right]^2.
\eeqa
\end{subequations}

Similar to the MMSE case,
the fixed point of these equations precisely agree with the
equations for replica analysis of MAP estimation with AWGN output noise
given in \cite{RanganFG:09arXiv} and related works in \cite{KabashimaWT:09arXiv}.
Thus, again, the GAMP framework provides a rigorous
justification of the replica results along with an algorithm
to achieve the predictions by the replica method.

\subsection{General Output Channels with Matched Distributions}
\label{sec:SEsumProdMatch}

Finally, let us return to the case of general (possibly non-AWGN) output channels
and consider the sum-product GAMP algorithm with the
estimations functions in Section \ref{sec:bpApproxMmse}.
In this case, we will assume that the postulated distributions for $p_{X|Q}$
and $p_{Y|Z}$ match the true distributions.
Guo and Wang in \cite{GuoW:07}
derived SE equations for standard BP in this case for
large sparse random matrices.
The work \cite{Rangan:10arXiv} derived identical equations for a relaxed version of BP.
We will see that the same SE equations will follow as a special case
of the more general SE equations above.

To prove this, we will show by induction that, for all $t$, $\Kbf^x(t)$
has the form
\beq \label{eq:KxMatch}
    \Kbf^x(t) = \left[\begin{array}{cc}
        \tau^{x0} & \tau^{x0} - \taubar^{x}(t) \\
        \tau^{x0} - \taubar^{x}(t) & \tau^{x0} - \taubar^{x}(t)
        \end{array} \right],
\eeq
where $\tau^{x0}$ is the variance of $X$.
For $t=0$, recall that for MMSE estimation, $\taubar^x(0) = \tau^{x0}$
and $\xhat_j(0)=0$ (the mean of $X$) for all $j$.
Therefore, \eqref{eq:SEinit} shows that \eqref{eq:KxMatch}
holds for $t=0$.

Now suppose that \eqref{eq:KxMatch} holds for some $t$.
Then, \eqref{eq:outSEKp} shows that $\Kbf^p(t)$ has the form
\beq \label{eq:KpMatch}
    \Kbf^p(t) = \left[\begin{array}{cc}
        \tau^{z0} & \tau^{z0} - \taubar^{p}(t) \\
        \tau^{z0} - \taubar^{p}(t) & \tau^{z0} - \taubar^{p}(t)
        \end{array} \right],
\eeq
where $\tau^{z0} = \beta \tau^{x0}$.  Now, if
$(Z,\Phat) \sim {\cal N}(0,\Kbf^p(t))$ with $\Kbf^p(t)$ given by \eqref{eq:KpMatch},
then the conditional distribution of $Z$ given $\Phat$ is
$Z \sim {\cal N}(\Phat, \taubar^p(t))$.
Therefore,  $\gout(\cdot)$ in \eqref{eq:goutMmseScore}
can be interpreted as
\beq \label{eq:goutMmseMatch}
    \gout(\phat,y,\tau^p) := \frac{\partial}{\partial \phat} \log
    p_{Y|\Phat}(y|\Phat=\phat),
\eeq
where $p_{Y|\Phat}(\cdot|\cdot)$ is the likelihood of $Y$
given $\Phat$ for the random variables $\theta^p(\Kbf^p(t))$
in \eqref{eq:thetaP}.

Now let $F(\taubar^p(t))$ be the \emph{Fisher information}
\beq \label{eq:FisherInfo}
    F(\taubar^p(t)) := \Exp\left[ \frac{\partial}{\partial \phat} \log
        p_{Y|\Phat}(Y|\Phat) \right]^2,
\eeq
where the expectation is also over $\theta^p(\Kbf^p(t))$
with $\Kbf^p(t)$ given by \eqref{eq:KpMatch}.
A standard property of the Fisher information is that
\beq \label{eq:FisherInfoB}
    F(\taubar^p(t)) = \Exp\left[ \frac{\partial^2}{\partial \phat^2} \log
    p_{Y|\Phat}(Y|\Phat) \right].
\eeq
Substituting \eqref{eq:FisherInfo} and \eqref{eq:FisherInfoB} into \eqref{eq:outSE}
we see that
\beq \label{eq:outSEMatch}
    \taubar^r(t) = \xi^r(t) = \frac{1}{F(\taubar^p(t))}.
\eeq

We will show in Appendix \ref{sec:alphaMatch} that
\beqa
    \lefteqn{ \Exp\left[ \frac{\partial}{\partial \phat}
    \gout(t,h(Z,W),\Phat,\taubar^p(t)) \right]} \nonumber \\
    &=&
     \Exp\left[ \frac{\partial}{\partial z}
    \gout(t,h(Z,W),\Phat,\taubar^p(t)) \right].\label{eq:dpzMatch}
\eeqa
It then follows from \eqref{eq:alphaSE} that
\beq \label{eq:alphaMatch}
    \alpha^r(t) = 1.
\eeq

Now consider the input update \eqref{eq:inSE}.
Since $\alpha^r(t)=1$ and $\taubar^r(t)=\xi^r(t)$, the expectation in
\eqref{eq:ginMmse} agrees with the expectations in \eqref{eq:inSE}.
Substituting \eqref{eq:ginMmseDeriv} into \eqref{eq:inSEmux},
we obtain
\beq \label{eq:muxMatch}
    \taubar^x(\tp1) = \Exp\left[ \var(X|Q,\Rhat) \right],
\eeq
where the expectation is over $\theta^r(\xi^r(t),\alpha^r(t))$ in
\eqref{eq:thetaR} with $\alpha = 1$.
Also, if $\Xhat(\tp1) = \Exp(X|Q,\Rhat)$ then
\beqan
    \lefteqn{ \Exp(X-\Xhat(\tp1))^2 = \Exp\left[ \var(X|Q,\Rhat) \right] = \taubar^x(\tp1) } \\
    & & \Exp\left[ \Xhat(\tp1)(X-\Xhat(\tp1)) \right] = 0.
    \hspace{1cm}
\eeqan
These relations, along with the definition $\tau^{x0} = \Exp(X^2)$,
show that the covariance $\Kbf^x(\tp1)$ in \eqref{eq:inSEKx} is of the form
\eqref{eq:KxMatch}.
This completes the induction argument to show that \eqref{eq:KxMatch} and
the other equations above hold for all $t$.

The updates \eqref{eq:outSEMatch} and \eqref{eq:muxMatch} are
precisely the SE equations given in \cite{GuoW:07} and \cite{Rangan:10arXiv}
for sum-product BP estimation with matched distributions on large sparse
random matrices.
The current work thus shows that the identical SE equations
hold for dense matrices $\Abf$.
In addition, the results here extend the earlier SE equations
by considering the cases where the distributions postulated by the estimator
may not be identical to the true distribution.

\section{Nonlinear Compressed Sensing Example} \label{sec:sparseNL}

\begin{figure}
	\begin{minipage}[b]{1.0\linewidth}
		\centering
		\includegraphics[width=7.5cm]{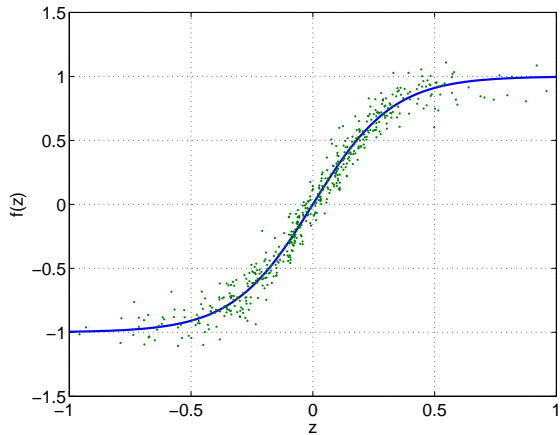}
	\end{minipage}
	\caption{Sigmoidal output function for the nonlinear compressed sensing example,
    along with a scatter plot of the noisy points $(z_i,y_i)$ in one typical realization.
    }
 	\label{fig:sparseNLOut}
\end{figure}

\begin{figure}
	\begin{minipage}[b]{1.0\linewidth}
		\centering
		\includegraphics[width=7.5cm]{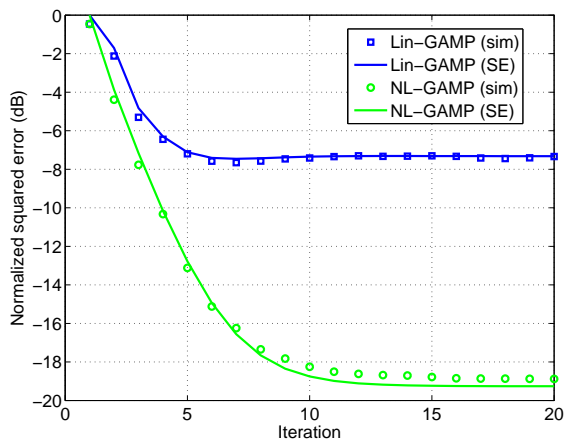}
	\end{minipage}
	\caption{Normalized squared error of the of the sparse vector estimates
   in the nonlinear compressed sensing
    problem.  Plotted are the median squared errors for Monte Carlo trials 
    of the GAMP algorithm
    using both a linear approximation of the output (Lin-GAMP) and the true nonlinear output 
    (NL-GAMP).  Also plotted are state evolution (SE) predictions of the asymptotic performance.
    }
 	\label{fig:sparseNLSim}
\end{figure}

To validate the GAMP algorithm and its SE analysis, we considered the following
simple example of a nonlinear compressed sensing problem.  
The distribution on the components of the input, $\xbf$, is taken as 
a Bernoulli-Gaussian:
\beq \label{eq:pxSim}
    x_j \sim \left\{ \begin{array}{ll}
        0 & \mbox{prob} = 1-\rho, \\
        {\mathcal N}(0,1) & \mbox{prob} = \rho,
        \end{array} \right.
\eeq
where, in the experiments below, we set $\rho=0.1$.
The distribution \eqref{eq:pxSim}
provides a simple model of a sparse vector with $\rho$ being the fraction of non-zero
components.  To match the state evolution (SE) theory, 
the linear transform $\Abf \in \R^{m\x n}$ is generated as an i.i.d.\ Gaussian matrix with
$A_{ij} \sim {\mathcal N}(0,1/n)$.
For the output channel, we consider nonlinear measurements of the form
\beq \label{eq:ysim}
    y_i = f(z_i) + w_i, \qquad f(z) = 1/(1+\exp(-az)),
\eeq
and $w_i \sim {\mathcal N}(0,\tau^w)$.
The function $f(z)$ in \eqref{eq:ysim} is a sigmoidal function that arises
commonly in neural networks and provides a good test of the ability of GAMP
to handle nonlinear measurements.
The output variance is taken as $\tau^w = 0.01$, or 20 dB below the full range of the
output of $f(z)$.  The scale factor  in \eqref{eq:ysim} is set to $a=6.1$.
The output function $f(z)$ along with a scatter plot of the noisy
points $(z_i,y_i)$ is shown in Fig.~\ref{fig:sparseNLOut}.
Since we are recovering a sparse vector $\xbf$ from a linear transform followed
by noisy, nonlinear measurements, we can consider the problem a nonlinear
compressed sensing problem.
We set $(m,n)=(500,1000)$ so that the sparse vector is undersampled by a factor of two.

For this problem, we consider two estimators, both based on the GAMP algorithm:
\begin{itemize}
\item \emph{NL-GAMP:}  The sum-product GAMP algorithm matched to the true input
and output distributions.
\item \emph{Lin-GAMP:} The sum-product GAMP algorithm matched to the true input
distribution, but the estimator assumes a linear channel of the form
\[
    y_i = f'(0)z_i + w_i, \qquad w_i \sim {\mathcal N}(0,\tau^w).
\]
\end{itemize}
Since the estimator Lin-GAMP assumes an AWGN output channel, it is equivalent to the 
Bayesian-AMP estimator of \cite{DonohoMM:10-ITW1,DonohoMM:10-ITW2}.
For both algorithms, we use the full algorithm, Algorithm \ref{algo:GAMP}.
Other simulations, not reported here, show virtually no difference to the 
simplified algorithm, Algorithm \ref{algo:GAMP-simp}.

Fig.~\ref{fig:sparseNLSim} plots the normalized squared error
for both algorithms over 100 Monte Carlo simulations.
In each Monte Carlo simulation we computed the normalized squared error (NSE),
\[
    \mbox{NSE} = 10\log_{10}(\Exp\|\xbf-\xbfhat\|^2/\Exp\|\xbf\|^2),
\]
where the expectation is over the components of the vectors.
Fig.~\ref{fig:sparseNLSim} then plots the median normalized
squared error over the 100 trials.
The median is used since there is a significant variation from  between trials,
and the median removes outliers.

The first point to notice is that there is a significant gain from
incorporating the nonlinearity in the output relative to using
simple linear approximations.  Indeed, NL-GAMP shows an asymptotic gain
of over 11 dB relative to the Lin-GAMP method that approximates the output as a
linear function.  Secondly, we see that both algorithms converge very quickly,
within approximately 10 to 15 iterations.

Finally, we see that for both methods, the state evolution (SE) analysis predicts
the per iteration performance extremely well.  For Lin-GAMP, the SE predicts the median
SE within 0.2 dB and for NL-GAMP, the prediction is within 0.4 dB.
Although the SE analysis theoretically only provides predictions for a modified version of 
the simplified Algorithm \ref{algo:GAMP-simp}, we see that the prediction holds,
at least in this simulation, for the full algorithm, Algorithm \ref{algo:GAMP}.

Also note that since NL-GAMP is the sum-product GAMP algorithm where the postulated
and true distributions are matched,
the SE equations fall under the special case given in Section \ref{sec:SEsumProdMatch}.
The SE equations in this special case 
were derived for sparse matrices in \cite{GuoW:07} and \cite{Rangan:10arXiv}.
However, since the Lin-GAMP is applied to a nonlinear output where the true and postulated
distributions are not matched, there are no previous SE analyses that can predict the
performance of the method.
Interestingly, the squared error for Lin-GAMP does not monotonically decrease;
There is a slight increase in squared error after iteration 7 
and the increase is actually predicted by the SE analysis.

The simulation thus demonstrates that the GAMP method can provide a tractable approach
to a difficult nonlinear compressed sensing problem with significant gains over 
AMP methods based on na\"ive linear approximations.  Moreover, the SE analysis
can precisely predict the performance of the method and quantify the value of
incorporating the nonlinearity.
All code for the simulation is available in the sourceforge GAMP repository
\cite{GAMPSourceForge}.

\section*{Conclusions and Future Work}

We have considered a general linear mixing estimation problem of 
estimating an i.i.d.\ (possibly non-Gaussian) vector observed through
a linear transform followed by a componentwise (possibly random and nonlinear)
measurements.  
The formulation is extremely general and
encompasses many problems in compressed sensing, classification,
learning and signal processing.
We have presented a novel algorithm, called GAMP, that unifies several earlier
methods to realize computationally efficient and systematic
approximations of max-sum and sum-product loopy BP.
Our main theoretical contribution is an extension of Bayati and
Montanari's state evolution (SE) analysis in \cite{BayatiM:11} 
to precisely describes the asymptotic behavior of the GAMP algorithm for
large Gaussian i.i.d.\ matrices.  
The GAMP algorithm thus provides a computationally simple
method that can apply to a large class of estimation problems with provable
exact performance characterizations.  As a result, we believe that the
GAMP algorithm can have wide ranging applicability.
Indeed, applications of the GAMP methodology have been used in 
nonlinear wireless scheduling problems \cite{RanganM:12},
compressed sensing with quantization \cite{KamilovGR:11arXiv},
and neural estimation \cite{FletcherRVB:11} to name a few.

Nevertheless, there are a large number of open issues for future research,
some of which have been begun consideration since the original publication
of this paper in \cite{Rangan:10arXiv-GAMP}.
\paragraph*{Non-Gaussian matrices}
Our SE analysis currently only applies to Gaussian i.i.d.\ matrices and a
significant open question is whether the analyses can be extended to more 
general matrices.  Recently, \cite{CaireSTV:11} extended the replica analysis in
\cite{Tanaka:02,GuoV:05} to obtain predictions of
the behavior of optimal estimators for linear mixing problems with general free matrices.
One avenue of future research is to see if a AMP-like algorithm can be constructed
that can provably obtain the performance predicted for these matrices.

\paragraph*{Learning Distributions}  One of the main limitations of the GAMP 
method is that both the sum-product and max-sum variants require that the true distributions
on the input an output channels are known.
When the distributions are not known, one approach is to find a
minimax estimator over a class of distributions, as was performed for classes of
sparse priors in \cite{DonJohMM:11}.   A second approach is to attempt to adaptively 
estimate the distributions assuming some parametric model.  
One of the most promising methods in this regard is to combine GAMP with expectation-maximization (EM) estimation as discussed in 
\cite{KrzMSSZ:11-arxiv,KrzMSSZ:12-arxiv,VilaSch:11,VilaSch:12,KamRanFU:12-arXiv}.

\paragraph*{Connections to graphical models}  Since the GAMP method is derived from
graphical model techniques, it can be incorporated as a component of a larger
graphical model where the approximate message passing is combined with standard
belief propagation updates.  Such hybrid approaches have been explored in a number 
of works including \cite{Schniter:10-CISS,ZinielPS:10,SomPS:10,Schniter:11,RanganFGS:12-ISIT}
for extending compressed sensing problems with other statistical dependencies 
between components or estimation of additional latent variables.

\paragraph*{Optimality}
A significant outstanding theoretical issue is to obtain a performance lower bound.
An appealing feature of the SE analysis of sparse random matrices
such as \cite{BoutrosC:02,GuoW:06,GuoW:07,Rangan:10arXiv} as well
as well as the replica analysis in \cite{Tanaka:02,GuoV:05,RanganFG:12-IT}
is that they provide lower bounds on the performance of the optimal estimator.
These lower bounds are also described by SE equations.  Then, using a sandwiching
argument -- a technique used commonly in the study of LDPC codes
\cite{RichardsonU:98} -- shows that
when fixed points of the SE equations are unique, BP is optimal.
Finding a similar lower bound for the
GAMP algorithm is a possible avenue of future research.

\section*{Acknowledgments}
The author would like to thank a number of people for their careful reading
of earlier versions of this paper as well as their contributions to 
the open-source GAMP software library \cite{GAMPSourceForge}: Alyson Fletcher,
Dongning Guo, Vivek K.\ Goyal, Andrea Montanari, Jason Parker, Phil Schniter, Amin Shokrollahi,
Lav Varshney, Jeremy Vila and Justin Ziniel.

\appendices

\section{Vector-Valued AMP} \label{sec:vecAMP}

\subsection{SE Equations for Vector-Valued AMP}

The analysis of the GAMP requires a vector-valued
version of the recursion analyzed by Bayati and Montanari.
Fix dimensions $n_d$ and $n_b$, and let $\Theta^u$ and $\Theta^v$ be
any two sets.
Let $\Gin(t,\dbf,\theta^u) \in \R^{n_b}$ and
$\Gout(t,\bbf,\theta^v) \in \R^{n_d}$ be vector-valued
functions over the arguments $t=0,1,2,\ldots$, $\bbf \in \R^{n_b}$,
$\dbf \in \R^{n_d}$, $\theta^u \in \Theta^u$ and $\theta \in \Theta^v$.
Assume $\Gin(t,\dbf,\theta^u)$ and $\Gout(t,\bbf,\theta^v)$
are Lipschitz continuous in $\dbf$ and $\bbf$, respectively.
Let $\Abf \in \R^{m \x n}$ and generate a sequence of
vectors $\bbf_i(t)$ and $\dbf_j(t)$ by the iterations
\begin{subequations} \label{eq:genAMPRec}
\beqa
    \bbf_i(t) &=& \sum_j a_{ij} \ubf_j(t) - \lambdabf(t)\vbf_i(\tm1),
        \label{eq:AMPGenb} \\
    \dbf_j(t) &=& \sum_i a_{ij}\vbf_i(t) - \xibf(t)\ubf_j(t) \label{eq:AMPGend}
\eeqa
\end{subequations}
where
\begin{subequations} \label{eq:uvAMPGen}
\beqa
    \ubf_j(\tp1) &=& \Gin(t,\dbf_j(t),\theta^u_j), \\
    \vbf_i(t) &=& \Gout(t,\bbf_i(t),\theta^v_i),
\eeqa
\end{subequations}
and
\begin{subequations} \label{eq:genAMPDeriv}
\beqa
    \xibf(t) &=& \frac{1}{m} \sum_{i=1}^m  \frac{\partial}{\partial \bbf}
    \Gout(t,\bbf_i(t),\theta^v_i)\\
    \lambdabf(\tp1) &=& \frac{1}{m} \sum_{j=1}^n
    \frac{\partial}{\partial \dbf}\Gin(t,\dbf_j(t),\theta^v_j).
        \label{eq:genAmpLam}
\eeqa
\end{subequations}
Here, we interpret the derivatives as matrices $\xibf(t)  \in \R^{n_d \x n_b}$
and $\lambdabf(\tp1) \in \R^{n_b \x n_d}$.
The recursion is initialized with $t=0$, $\vbf_i(\tm1) = 0$ and
some values for $\ubf_j(0)$.

Now similar to Section \ref{sec:asymAnal}, consider
a sequence of random realizations of the
parameters indexed by the input dimension $n$.  For each $n$, we
assume that the output dimension $m=m(n)$
is deterministic and scales linearly as in \eqref{eq:betaDef}
for some $\beta \geq 0$.
Assume that the transform matrix $\Abf$ has i.i.d.\ Gaussian
components $a_{ij} \sim {\cal N}(0,1/m)$.
Also assume that the following components converge empirically
with bounded moments of order $2k-2$ with the limits
\begin{subequations} \label{eq:thetaUVLim}
\beqa
    \lim_{n \arr \infty} \theta_i^v &\stackrel{d}{=}& \theta^v, \ \ \
    \lim_{n \arr \infty} \theta_j^u \stackrel{d}{=} \theta^u, \\
    \lim_{n \arr \infty} \ubf_i(0) &\stackrel{d}{=}& \Ubf_0,
\eeqa
\end{subequations}
for random variables $\theta^v$, $\theta^u$ and $U_0$.
We also assume $\Ubf_0$ is zero mean.

Under these assumptions, we will argue that the SE equations for the vector-valued
recursion are given by:
\begin{subequations} \label{eq:SEGen}
\beqa
    \lefteqn{ \Kbf^d(t) } \nonumber \\
    &:=& \Exp\left[ \Gout(t, \Bbf(t), \theta^v)
        \Gout(t, \Bbf(t), \theta^v)^T \right] \label{eq:Kd} \\
    \lefteqn{ \Kbf^b(\tp1) } \nonumber \\
   &:=& \beta \Exp\left[ \Gin(t, \Dbf(t), \theta^u)
        \Gin(t, \Dbf(t), \theta^u)^T \right] \label{eq:Kb}
\eeqa
\end{subequations}
where the expectations are over the random variables $\theta^u$ and $\theta^v$
in the limits \eqref{eq:thetaUVLim}, and $\Bbf(t)$ and $\Dbf(t)$ are Gaussian
vectors
\beq \label{eq:BDGauss}
    \Bbf(t) \sim {\cal N}(0,\Kbf^b(t)), \ \
    \Dbf(t) \sim {\cal N}(0,\Kbf^d(t))
\eeq
independent of $\theta^u$ and $\theta^v$.
The SE equations \eqref{eq:SEGen} are initialized with
\beq \label{eq:SEGenInit}
    \Kbf^b(0) := \beta \Exp\left[ \Ubf_0\Ubf_0^T \right].
\eeq

The matrices in \eqref{eq:genAMPDeriv} are derived \emph{empirically} from
the variables $\bbf(t)$ and $\dbf(t)$.  We will also be interested
in the case where the algorithm directly uses the \emph{expected} values
\begin{subequations} \label{eq:genAMPDerivExp}
\beqa
    \xibf(t) &=& \Exp\left[ \frac{\partial}{\partial \bbf}
    \Gout(t,\Bbf(t),\theta^v) \right] \\
    \lambdabf(\tp1) &=& \Exp\left[
    \frac{\partial}{\partial \dbf}\Gin(t,\Dbf(t),\theta^v) \right],
\eeqa
\end{subequations}
where, again, the expectations are over
random variables $\theta^u$ and $\theta^v$
in the limits \eqref{eq:thetaUVLim}, and $\Bbf(t)$ and $\Dbf(t)$ are Gaussian
vectors in \eqref{eq:BDGauss} independent of $\theta^u$ and $\theta^v$.

\begin{claim} \label{thm:SEGen}
Consider the recursion in \eqref{eq:genAMPRec} and \eqref{eq:uvAMPGen} with
either the empirical update \eqref{eq:genAMPDeriv} or expected
update \eqref{eq:genAMPDerivExp}.
Then, under the above assumptions, for any fixed
iteration number $t$, the variables in the recursions with either the
empirical or expected update converge empirically as
\begin{subequations} \label{eq:vecAMPLim}
\beqa
    \lim_{n \arr \infty} (\dbf_j(t),\theta^u_j)
        &\stackrel{d}{=}& (\Dbf(t),\theta^u) \\
    \lim_{n \arr \infty} (\bbf_i(t),\theta^v_i) &\stackrel{d}{=}&
        (\Bbf(t),\theta^v),
\eeqa
\end{subequations}
where are $\theta^u$ and $\theta^v$ are the random variables
in the limits \eqref{eq:thetaUVLim}, and $\Bbf(t)$ and $\Dbf(t)$ are Gaussians
\eqref{eq:BDGauss} independent of $\theta^u$ and $\theta^v$.
\end{claim}

The scalar case when $n_b=n_d=1$ is rigorously proven by Bayati and Montanari
\cite{BayatiM:11}.  The modifications for the vector-valued case
is straightforward but tedious.  
However, we only provide a sketch of the arguments in Appendix \ref{sec:SEGenPf}.
As discussed above, a full re-derivation of the proof from \cite{BayatiM:11}
would be long and beyond the scope of the paper.
Thus, the result is not fully rigorous, and we use the term Claim instead of
Theorem to emphasize the lack of rigor.
A complete proof would be a valuable avenue of future work.

\section{Proof of Claim~\ref{thm:gampSE}}
Claim~\ref{thm:gampSE} is a special case of the general result, Claim~\ref{thm:SEGen}, above.
Although we have not provided a complete proof of Claim~\ref{thm:SEGen}, 
the implication from Claim~\ref{thm:SEGen} to \ref{thm:gampSE} is completely rigorous.

Let $n_b=2$ and $n_d=1$ and define the variables
\begin{subequations} \label{eq:ampEquiv}
\beqa
    \ubf_j(t) &=& \left[\begin{array}{c} x_j \\ \xhat_j(t) \end{array} \right], \ \
    \bbf_i(t) = \left[ \begin{array}{c} z_i \\ \phat_i(t) \end{array} \right]
        \label{eq:ampbu}\\
    v_i(t) &=& \shat_i(t) \\
    d_j(t) &=& \frac{1}{\taubar^r(t)}(\rhat_j(t) - \alpha^r(t)x_j)
        \label{eq:ampdj} \\
    \theta_j^u &=& (x_j,q_j), \ \ \
    \theta_i^v = w_i \\
    \lambdabf(\tp1) &=& \left[ \begin{array}{c}0 \\ \tau^p(t) \end{array}\right] \label{eq:lamEquiv}\\
    \xibf(t) &=& \frac{1}{\taubar^r(t)}\left[ \alpha^r(t) \  -1 \right]. \label{eq:xiEquiv}
\eeqa
\end{subequations}
Also, for $\theta^u = (x,q)$, $\theta^v = w$ and $\bbf = (z \ \phat)^T$,
define the functions
\begin{subequations} \label{eq:Gequiv}
\beqa
    \lefteqn{\Gin(t,d,\theta^u)} \nonumber \\
    &:=& \left[ \begin{array}{c} x \\
    \gin(t,\taubar^r(t)d+\alpha^r(t)x, q, \taubar^r(t)) \end{array}\right] ,
        \label{eq:Ginequiv}
\eeqa
and
\beq \label{eq:Goutequiv}
    \Gout(t,\bbf,\theta^v) := \gout(t,\phat, h(z,w), \taubar^p(t)).
\eeq
\end{subequations}
With these definitions, it is easily checked that simplified
GAMP algorithm, Algorithm \ref{algo:GAMP-simp}, 
with the modifications in Section~\ref{sec:mainResult} agrees with
the recursion described by equations \eqref{eq:genAMPRec},
and \eqref{eq:uvAMPGen}, with $\lambdabf(t)$ being defined with the empirical update in
\eqref{eq:genAMPDeriv} and $\xibf(t)$ being defined by the expected value in
\eqref{eq:genAMPDerivExp}.
For example,
\beqan
    \bbf_i(t) &\stackrel{(a)}{=}& \left[ \begin{array}{c}z_i \\ \phat_i(t)
        \end{array} \right] \nonumber \\
    &\stackrel{(b)}{=}& \sum_j a_{ij}
            \left[ \begin{array}{c}x_j \\ \xhat_j(t)\end{array} \right]  -
            \left[ \begin{array}{c}0 \\ \taubar^p(t)\shat_i(\tm1) \end{array} \right]
            \nonumber \\
    &\stackrel{(c)}{=}&
        \sum_j a_{ij}\ubf_j(t) - \lambdabf(t)v_i(\tm1),
\eeqan
where (a) follows from the definition of $b_i(t)$ in \eqref{eq:ampbu};
(b) follows and \eqref{eq:outLinSimp} and the fact that $\zbf = \Abf\xbf$
and (c) follows from the remaining definitions in \eqref{eq:ampEquiv}.
Therefore,  the variables in \eqref{eq:ampEquiv} satisfy \eqref{eq:AMPGenb}.
Similarly,
\beqan
    d_j(t) &\stackrel{(a)}{=}&
        \frac{1}{\taubar^r(t)}(\rhat_j(t) - \alpha^r(t)x_j) \nonumber \\
        &\stackrel{(b)}{=}&
        \sum_i a_{ij}\shat_i(t) + \frac{1}{\taubar^r(t)}
        (\xhat_j(t) - \alpha^r(t)x_j) \nonumber \\
        &\stackrel{(c)}{=}&
        \sum_i a_{ij}v_i(t) - \xibf(t)\ubf_j(t),
\eeqan
where (a) follows from the definition of $d_j(t)$ in \eqref{eq:ampdj};
(b) follows from \eqref{eq:rhatSimp} with the modification that
$\tau^r(t)$ is replaced with $\taubar^r(t)$;
and (c) follows from the other definitions in \eqref{eq:ampEquiv}
Hence the variables in \eqref{eq:ampEquiv} satisfy \eqref{eq:AMPGend}.
The equations in \eqref{eq:uvAMPGen} can also easily verified.

We next consider $\lambdabf(t)$ and $\xibf(t)$.
For $\lambdabf(t)$, first observe that \eqref{eq:ampdj} shows that
\beqa
    \lefteqn{\frac{\partial}{\partial d} \left.
        \gin(t,\taubar^r(t)d+\alpha^r(t)x_j, q_j, \taubar^r(t))\right|_{d=d_j(t)} } 
            \nonumber \\
    &=& \taubar^r(t)\frac{\partial}{\partial \rhat}
    \gin(t,\rhat_j(t), q_j, \taubar^r(t)). \label{eq:ginDerivA}
\eeqa
Hence 
\beqa
    \lefteqn{\tau^p(\tp1) \stackrel{(a)}{=} \frac{n}{m}\tau^x(\tp1) } \nonumber \\
    = \frac{1}{n}\sum_{j=1}^n \taubar^r(t)
    \frac{\partial}{\partial \rhat} \gin(t,\rhat_j(t), q_j, \taubar^r(t)) 
        \label{eq:ginDerivB}
\eeqa
where (a) follows from \eqref{eq:mupSimp} with the modification that 
$\|\Abf\|^2_F$ has been replaced by its expectation $\Exp\|\Abf\|^2_F = n$ and
(b) follows from \eqref{eq:muxSimp} with the modification that $\tau^r(t)$
is replaced by $\taubar^r(t)$.  Combining \eqref{eq:ginDerivA} and \eqref{eq:ginDerivB},
with the definition of $\Gin(\cdot)$ in \eqref{eq:Ginequiv},
we see that
$\lambdabf(t)$ in \eqref{eq:lamEquiv} satisfies \eqref{eq:genAmpLam}.

Similarly, for $\xibf(t)$, observe that \eqref{eq:alphaSE} and \eqref{eq:outSEmur} show
that
\beqan
    \Exp\left[ \left. \frac{\partial}{\partial z}
        \gout(t,\Phat,h(z,W),\taubar^p(t)) \right|_{z=Z} \right] &=& \frac{\alpha^r(t)}{\taubar^r(t)} \\
    \Exp\left[ \left. \frac{\partial}{\partial \phat}
        \gout(t,\phat,h(Z,W),\taubar^p(t)) \right|_{\phat=P} \right] &=& \frac{-1}{\taubar^r(t)}.
\eeqan
Combining these relations with the definition of $\Gout(\cdot)$ in \eqref{eq:Goutequiv},
shows that $\xibf(t)$ defined in \eqref{eq:xiEquiv} satisfies \eqref{eq:genAMPDerivExp}.

Therefore, we can apply Claim~\ref{thm:SEGen} which shows that the limits
\eqref{eq:vecAMPLim} hold for the matrices $\Kbf^b(t)$ and $\Kbf^d(t)$
from the SE equations \eqref{eq:SEGen} with initial condition
\eqref{eq:SEGenInit}.  Since the definitions in \eqref{eq:ampEquiv}
set $\theta^u_j = (x_j,q_j)$ and $\theta^v_i = w_i$, the
expectations in \eqref{eq:SEGen} are over the random variables
\beq \label{eq:ThetaUVEquiv}
    \theta^u \sim (X,Q), \ \ \ \theta^v \sim W,
\eeq
where $(X,Q)$ and $W$ are the limiting random variables in Section
\ref{sec:assumptions}.

Now using the SE equations \eqref{eq:SEinit},
\eqref{eq:outSE} and \eqref{eq:inSE}, one can easily show
by induction that
\begin{subequations} \label{eq:ampSEEquiv}
\beqa
    \Kbf^p(t) &=& \Kbf^b(t) = \beta \Kbf^x(t) \label{eq:KpEquiv} \\
    \xi^r(t) &=& (\taubar^r(t))^2\Kbf^d(t). \label{eq:taurEquiv}
\eeqa
\end{subequations}
For example, one can show that
\eqref{eq:KpEquiv} holds for $t=0$ by comparing the
initial conditions \eqref{eq:SEGenInit} with \eqref{eq:SEinit}
and using the definition of $\ubf(0)$ in \eqref{eq:ampbu}.
Now, suppose that \eqref{eq:KpEquiv} holds or some $t$.
Then, \eqref{eq:KpEquiv} and \eqref{eq:ThetaUVEquiv}
show that
$(\Bbf(t),\theta^v)$ in the expectation \eqref{eq:Kb}
is identically distributed to
$((Z,\Phat), W)$ in $\theta^p(\Kbf^p(t))$.
Therefore,
\beqan
    \xi^r(t) &\stackrel{(a)}{=}&
        (\taubar^r(t))^2\Exp\left[ \gout^2(t,Y,\Phat,\taubar^p(t)) \right]
        \nonumber \\
    &\stackrel{(b)}{=}&
        (\taubar^r(t))^2\Exp\left[ \Gout^2(t,\Bbf(t), \theta^v) \right]
        \nonumber \\
    &\stackrel{(c)}{=}& (\taubar^r(t))^2\Kbf^d(t), \nonumber
\eeqan
where (a) follows from \eqref{eq:taurSE};
(b) follows from the definition of $\bbf(t)$ and $\theta^v$ in \eqref{eq:ampEquiv}
and $\Gout(\cdot)$ in \eqref{eq:Goutequiv}; and
(c) follows from \eqref{eq:Kd}.  Similarly, one can show that,
if \eqref{eq:taurEquiv} holds for some $t$, then
\eqref{eq:KpEquiv} holds for $t+1$.

With these equivalences, we can now prove the assertions in Claim~\ref{thm:gampSE}.
To prove \eqref{eq:thetarlim}, first observe that the limit \eqref{eq:vecAMPLim}
along with \eqref{eq:ThetaUVEquiv} and
the definitions in \eqref{eq:ampEquiv} show that
\[
  \lim_{n \arr \infty}
    \left(\frac{1}{\taubar^r(t)}(\rhat_j(t)-\alpha^r(t)x_j),
        x_j,q_j\right)  \stackrel{d}{=} (D(t),X,Q),
\]
where the limit is in the sense of empirical convergence of order $k$
and $D(t) \sim {\cal N}(0,\Kbf^d(t))$ is independent of $(X,Q)$.
But, this limit is equivalent to
\beq \label{eq:xqrLimPf}
    \lim_{n \arr \infty} (x_j,q_j,\rhat_j(t)) \stackrel{d}{=} (X,Q,\Rhat),
\eeq
where
\[
    \Rhat = \alpha^r(t)X + \taubar^r(t)D(t).
\]
Since $D(t)  \sim {\cal N}(0,\Kbf^d(t))$,
\[
    \taubar^r(t)D(t) \sim {\cal N}(0,(\taubar^r(t))^2\Kbf^d(t)) =
    {\cal N}(0,\xi^r(t)),
\]
where the last equality is due to \eqref{eq:taurEquiv}.
Therefore, $(X,Q,\Rhat)$ in \eqref{eq:xqrLimPf} is identically
distributed to $\theta^r(\xi^r(t),\alpha^r(t))$ in \eqref{eq:thetaR}.
This proves \eqref{eq:thetarlim}, and part (b) of Claim~\ref{thm:gampSE}.
Part (c) of Claim~\ref{thm:gampSE} is proven similarly.

To prove \eqref{eq:muLim} in part (a),
\beqan
    \lefteqn{ \lim_{n \arr \infty} \frac{1}{\tau^r(t)}
    \stackrel{(a)}{=} \lim_{n \arr \infty} \tau^s(t) }\nonumber \\
    &\stackrel{(b)}{=}& -\lim_{n \arr \infty}
    \frac{1}{m}\sum_{i=1}^m \frac{\partial}{\partial \phat}
        \gout(t,\phat_i(t),y_i,\taubar^p(t)) \nonumber \\
    &\stackrel{(c)}{=}& -\Exp\left[ \frac{\partial}{\partial \phat}
        \gout(t,\Phat,Y,\taubar^p(t)) \right] \nonumber \\
     &\stackrel{(d)}{=}&  \frac{1}{\taubar^r(t)},
\eeqan
where (a) follows from \eqref{eq:murSimp} with the modification that $\|\Abf\|_F^2=n$,
(b) follows from \eqref{eq:tausSimp} and that fact that we are considering
the modified algorithm where $\tau^p(t)$ is replaced with $\taubar^p(t)$;
(c) follows the limit \eqref{eq:thetaplim} and the assumption that
the derivative of $\gout(t,\phat,y,\taubar^p(t))$ is of order $k$;
and (d) follows from \eqref{eq:outSE}.
Similarly, one can show
\[
    \lim_{n \arr \infty} \tau^p(t) = \taubar^p(t).
\]
This proves \eqref{eq:muLim} and completes the proof of Claim~\ref{thm:gampSE}.

\section{Max-Sum GAMP} \label{sec:bpMap}

In this section, we show that with
the functions $\gin$ and $\gout$ in \eqref{eq:ginMap} and \eqref{eq:goutMap},
the GAMP algorithm can be seen heuristically as a first-order
approximation of max-sum loopy BP for the MAP estimation problem.
The derivations in this section are not rigorous, since we do not
formally claim any properties of this approximation.

\subsection{Max-Sum BP for MAP Estimation}
We first review how we would apply standard max-sum loopy BP for the MAP estimation
problem \eqref{eq:xhatMap}.
For both the MAP and MMSE estimation problems,
standard loopy BP operates by associating with the transform matrix $\Abf$
a bipartite graph $G=(V,E)$
called the \emph{factor} or \emph{Tanner} graph
as illustrated in Fig.~\ref{fig:BPGraph}.
The vertices $V$ in this graph consists of $n$ ``input" or
``variable" nodes
associated with the variables $x_j$, $j=1,\ldots,n$,
and $m$ ``output" or ``measurements" nodes
associated with the transform outputs
$z_i$, $i=1,\ldots,m$.
There is an edge $(i,j) \in E$ between the input node $x_j$
and output node $z_i$ if and only if $a_{ij} \neq 0$.

\begin{figure}
\begin{center}
  \includegraphics[width=2in,height=2in]{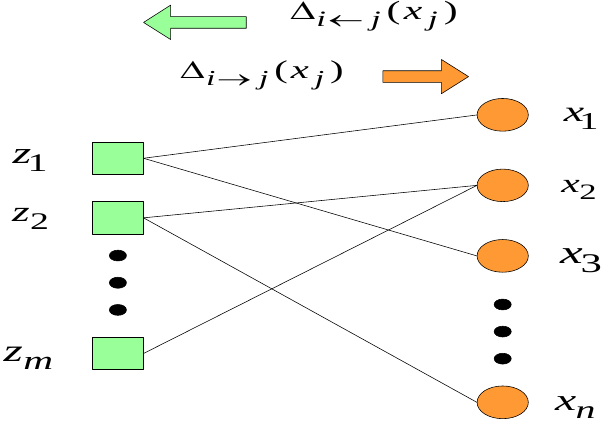}
\end{center}
\caption{Factor or Tanner graph for the linear mixing estimation problem. }
\label{fig:BPGraph}
\end{figure}

Now let $\Delta_j(x_j)$ be the marginal maxima in \eqref{eq:margUtil},
which we can interpret as a ``value" function on the variable $x_j$.
Max-sum loopy BP iteratively sends
messages between the input nodes $x_j$ and output nodes $z_i$
representing estimates of the value function $\Delta_j(x_j)$.
The value message from the input node $x_j$ to output node $z_i$ in the $t$th
iteration is denoted $\Delta_{i \la j}(t,x_j)$ and
the reverse message for $z_i$ to $x_j$ is denoted $\Delta_{i \ra j}(t,x_j)$.

Loopy BP updates the value messages with the following simple recursions:
The messages from the output nodes are given by
\beqa
    \lefteqn{\Delta_{i \ra j}(t,x_j) = \mbox{const} } \nonumber \\
        &+& \max_{\xbf} \fout(z_i,y_i)
        + \sum_{r \neq j} \Delta_{i \la r}(t,x_r)  \label{eq:outBPMap}
\eeqa
where the maximization is over vectors $\xbf$ with the $j$th component
equal to $x_j$ and $z_i = \abf_i^T\xbf$,
where $\abf_i^T$ is the $i$th row of the matrix $\Abf$.
The constant term is any term that does not depend on $x_j$, although
it may depend on $t$ or the indices $i$ and $j$.
The messages from the input nodes are given by
\beqa
   \lefteqn{ \Delta_{i \la j}(\tp1,x_j) =  \mbox{const} } \nonumber \\
   &+& \fin(x_j,q_j) +
        \sum_{\ell \neq i} \Delta_{\ell \ra j}(t,x_j) \label{eq:inBPMap}
\eeqa
The BP iterations are initialized with $t=0$ and
$\Delta_{i \ra j}(-1,x_j) = 0$.

The BP algorithm is terminated after some finite number of iterations.
After the final $t$th iteration, the final estimate for $x_j$ can
be taken as the maximum of
\beq \label{eq:inBPMapTot}
   \Delta_{j}(\tp1,x_j) = \fin(x_j,q_j) + \sum_{i} \Delta_{i \ra j}(t,x_j).
\eeq

\subsection{Quadratic Legendre Transforms} \label{sec:LegQuad}
To approximate the BP algorithm, we need the following simple result.
Given a function $f: \R \arr \R$, define the functions
\begin{subequations} \label{eq:lamgamDef}
\beqa
    (Lf)(x,r,\tau) &:=& f(x) - \frac{1}{2\tau}(r-x)^2 \\
    (\Gamma f)(r,\tau) &:=& \argmax_x (Lf)(x,r,\tau) \label{eq:GammaDef} \\
    (\Lambda f)(r,\tau) &:=& \max_x (Lf)(x,r,\tau)
        \label{eq:Lam0Def} \\
    (\Lambda^{(k)} f)(r,\tau) &:=& \frac{\partial^k}{\partial r^k} (\Lambda f)(r,\tau),
        \label{eq:LamkDef}
\eeqa
\end{subequations}
over the variables $r,\tau \in \R$ with $\tau > 0$ and $k=1,2,\ldots$.
The function $\Lambda f$ can be interpreted as a quadratic variant of
the standard Legendre transform of $f$ \cite{Rockafellar:70}.

\begin{lemma} \label{lem:lamDeriv}  Let $f:\R \arr \R$ be twice differentiable
and assume that all the maximizations in \eqref{eq:lamgamDef} exist and are unique.
Then,
\begin{subequations} \label{eq:lamgamDeriv}
\beqa
    (\Lambda^{(1)} f)(r,\tau) &=& \frac{\xhat-r}{\tau} \label{eq:Lam1Deriv} \\
    (\Lambda^{(2)} f)(r,\tau) &=& \frac{f''(\xhat)}{1-\tau f''(\xhat)},
         \label{eq:Lam2Derivb} \\
    \frac{\partial}{\partial r}\xhat &=&\frac{1}{1-\tau f''(\xhat)}
        \label{eq:GamDeriv}
\eeqa
where $\xhat = (\Gamma f)(r,\tau)$.
\end{subequations}
\end{lemma}
\begin{proof}
Equation \eqref{eq:Lam1Deriv} follows from the fact that
\[
    (\Lambda^{(1)} f)(r,\tau) =
    \left. \frac{\partial L(r,x)}{\partial r}\right|_{x = \xhat}
    = \frac{\xhat-r}{\tau}.
\]
Next observe that
\beqan
    \frac{\partial^2}{\partial x^2}L(x,r) &=& f''(x) - \frac{1}{\tau} \\
    \frac{\partial^2}{\partial x\partial r}L(x,r) &=& \frac{1}{\tau}.
\eeqan
So the derivative of $\xhat$ is given by
\beqan
    \lefteqn{\frac{\partial}{\partial r}\xhat
        = -\left[\frac{\partial^2}{\partial x^2}L(\xhat,r)\right]^{-1}
        \frac{\partial^2}{\partial x\partial r}L(\xhat,r) }\\
        &=& \frac{1/\tau}{1/\tau-f''(x)} = \frac{1}{1-\tau f''(x)},
\eeqan
which proves \eqref{eq:GamDeriv}.
Hence
\beqan
    \lefteqn{ (\Lambda^{(2)} f)(r,\tau) = \frac{\partial}{\partial r}
        (\Lambda^{(1)} f)(r,\tau) } \\
    &=& \frac{\partial}{\partial r} \left[ \frac{\xhat-r}{\tau} \right] \\
    &=& \frac{1}{\tau}\left(\frac{1}{ 1 - \tau f''(\xhat)}-1\right)
        = \frac{f''(\xhat)}{1-\tau f''(\xhat)},
\eeqan
which shows \eqref{eq:Lam2Derivb}.
\end{proof}

\subsection{GAMP Approximations}
We now show that the GAMP algorithm
with the functions in \eqref{eq:ginMap} and \eqref{eq:goutMap}
can be seen as a quadratic approximation
of the max-sum loopy BP updates \eqref{eq:inBPMap}.
The derivation is similar to the one given in \cite{Montanari:12-bookChap}
for the Laplacian AMP algorithm \cite{DonohoMM:09}.
We begin by considering the output node update \eqref{eq:outBPMap}.
Let
\begin{subequations} \label{eq:xhatMuMap}
\beqa
    \xhat_{j}(t) &:=& \argmax_{x_j} \Delta_{j}(t,x_j), \label{eq:xhatMapDelj} \\
    \xhat_{i \la j}(t) &:=& \argmax_{x_j} \Delta_{i \la j}(t,x_j),
        \label{eq:xhatMapDel} \\
    \frac{1}{\tau^x_{j}(t)} &:=& -\frac{\partial^2}{\partial x_j^2}
        \left. \Delta_{j}(t,x_j) \right|_{x_j = \xhat_{j}(t)}.
        \label{eq:muxMapDelj}  \\
    \frac{1}{\tau^x_{i \la j}(t)} &:=& -\frac{\partial^2}{\partial x_j^2}
        \left. \Delta_{i \la j}(t,x_j) \right|_{x_j = \xhat_{i \la j}(t)}.
        \label{eq:muxMapDel}
\eeqa
\end{subequations}
Now, for small $a_{ir}$, the values of $x_r$ in the maximization
\eqref{eq:outBPMap} will be close to $\xhat_{i \la j}(t)$.
So, we can approximate each term $\Delta_{i \la r}(t,x_r)$ with
the second order approximation
\beq \label{eq:DelMaxLQuad}
    \Delta_{i \la r}(t,x_r) \approx \Delta_{i \la r}(t,\xhat_{i\la r}(t))
         -\frac{1}{2\tau^x_r(t)}(x_r -\xhat_{i \la r}(t))^2,
\eeq
where we have additionally made the approximation $\tau^x_{i \la r}(t) \approx
\tau^x_r(t)$ for all $i$.
Now, given $x_j$ and $z_i$, consider the minimization
\beq \label{eq:minMAPQuad}
    J:= \min_{\xbf} \sum_{r \neq j} \frac{1}{2\tau_r}(x_r -\xhat_{i \la r})^2,
\eeq
subject to
\[
    z_i = a_{ij}x_j + \sum_{r \neq j} a_{ir}x_r.
\]
A standard least squares calculation shows that the minimization \eqref{eq:minMAPQuad}
is given by
\[
    J = \frac{1}{2} \sum_{r \neq j} \frac{1}{2\tau^p_{i \ra j}}
        \left(z_i-\phat_{i \ra j}-a_{ij}x_j\right)^2
\]
where
\[
    \phat_{i \ra j} = \sum_{r \neq j} a_{ir}\xhat_r, \ \
    \tau^p_{i \ra j} = \sum_{r \neq j} |a_{ir}|^2\tau^x_r.
\]
So, the approximation \eqref{eq:DelMaxLQuad} reduces to
\beqa
    \lefteqn{ \Delta_{i \ra j}(t,x_j) }\nonumber\\
     &\approx& \max_{z_i} \left[ \fout(z_i,y_i)
        - \frac{1}{2\tau^p_{i \ra j}(t)}(z_i-\phat_{i\ra j}(t)-a_{ij}x_j)^2,
        \right] \nonumber \\
        & & + \mbox{const} \nonumber \\
        &=& H\left(
            \phat_{i\ra j}(t)+a_{ij}x_j, y_i,
            \tau^p_{i \ra j}(t)\right) + \mbox{const} \label{eq:DelRH1}
\eeqa
where
\begin{subequations} \label{eq:zmuijMap}
\beqa
    \phat_{i \ra j}(t) &:=& \sum_{r \neq j} a_{ir}\xhat_{i\la r}(t) \\
    \tau^p_{i \ra j}(t) &:=& \sum_{r \neq j} |a_{ir}|^2\tau^x_r(t),
\eeqa
\end{subequations}
and
\beq \label{eq:Hmap}
    H(\phat,y,\tau^p) := \max_z \left[
        \fout(z,y) - \frac{1}{2\tau^p}(z-\phat)^2 \right].
\eeq
The constant term in \eqref{eq:DelRH1} does not depend on $z_i$.
Now let
\beq \label{eq:zhatiMap}
    \phat_i(t) := \sum_j a_{ij}\xhat_{i \la j}(t),
\eeq
and $\tau^p_i(t)$ be given as in \eqref{eq:mup}.
Then it follows from \eqref{eq:zmuijMap} that
\begin{subequations} \label{eq:zhatijdiff}
\beqa
    \phat_{i \ra j}(t) &=& \phat_i(t) - a_{ij}\xhat_{i \la j}(t) \\
    \tau^p_{i \ra j}(t) &=& \tau^p_i(t) - a_{ij}^2\tau^x_{j}(t).
\eeqa
\end{subequations}
Using \eqref{eq:zhatijdiff} and neglecting terms of order $O(a_{ij}^2)$, \eqref{eq:DelRH1} can be further approximated as
\beq \label{eq:DelRH2}
    \Delta_{i \ra j}(t,x_j) \approx H\left(
            \phat_{i}(t)+a_{ij}(x_j-\xhat_j), y_i,\tau^p_i(t)\right)
            + \mbox{const.}
\eeq
Now let
\beq \label{eq:goutMap1}
    \gout(\phat,y,\tau^p) := \frac{\partial}{\partial \phat}
        H(\phat,y,\tau^p).
\eeq
Comparing $H(\cdot)$  in \eqref{eq:Hmap} with the definitions
in \eqref{eq:lamgamDef} and using the properties in \eqref{eq:lamgamDeriv},
it can be checked that
\[
    H(\phat,y,\tau^p) =
        (\Lambda\fout(\cdot,y))\left(\phat,\tau^p\right),
\]
and that $\gout(\cdot)$ in \eqref{eq:goutMap1}
and its derivative agree with the definitions in
\eqref{eq:goutMap} and \eqref{eq:goutMapDeriv}.

Now define $\shat_i(t)$ and $\tau^s_i(t)$ as in \eqref{eq:outNL}.
Then, a first order approximation of \eqref{eq:DelRH2} shows that
\beqa
    \lefteqn{\Delta_{i \ra j}(t,x_j) \approx \mbox{const} } \nonumber \\
    & & +s_i(t)a_{ij}(x_j-\xhat_j(t))
        -\frac{\tau^s_i(t)}{2}a_{ij}^2(x_j-\xhat_j(t))^2. \nonumber \\
    &=& \mbox{const} \left[s_i(t)a_{ij}+a_{ij}^2\tau^s_i(t)\xhat_j(t)\right]x_j
        \nonumber\\
    & & - \frac{\tau^s_i(t)}{2}a_{ij}^2x_j^2, \label{eq:DelRAppMap}
\eeqa
where again the constant term does not depend on $x_j$.
We next consider the input update \eqref{eq:inBPMap}.
Substituting the approximation \eqref{eq:DelRAppMap} into
\eqref{eq:inBPMap} we obtain
\beqa
    \lefteqn{\Delta_{i \la j}(\tp1,x_j) \approx \mbox{const} } \nonumber \\
    &+& \fin(x_j,q_j) - \frac{1}{2\tau^r_{i \la j}(t)}
    (\rhat_{i \la j}(t)-x_j)^2,
        \label{eq:DelLApp1}
\eeqa
where the constant term does not depend on $x_j$ and
\begin{subequations} \label{eq:rmuijMap}
\beqa
    \frac{1}{\tau^r_{i \la j}(t)} &=&
        \sum_{\ell \neq i} a_{\ell j}^2\tau^s_\ell(t) \\
    \rhat_{i \la j}(t) &=& \tau^r_{i\la j}(t)\sum_{\ell \neq i}
        \left[ s_\ell(t)a_{\ell j}+
        a_{\ell j}^2\tau^s_\ell(t)\xhat_j(t) \right] \nonumber \\
        &=& \xhat_j(t) + \tau^r_{i\la j}(t)\sum_{\ell \neq i} s_\ell(t)a_{\ell j}.
\eeqa
\end{subequations}
Now define $\gin(r,q,\tau^r)$ as in \eqref{eq:ginMap}.
Then, using \eqref{eq:DelLApp1} and \eqref{eq:ginMap},
$\xhat_{i \la j}(t)$ in \eqref{eq:xhatMapDel} can be re-written as
\beq \label{eq:xhatijMap1}
    \xhat_{i \la j}(\tp1) \approx \gin\left(\rhat_{i \la j}(t), q_j, \tau^r_{i \la j}(t)\right).
\eeq
Then, if we define $\rhat_j(t)$ and $\tau^r_j(t)$ as in \eqref{eq:inLin},
$\rhat_{i \la j}(t)$ and $\tau^r_{i \la j}(t)$ in \eqref{eq:rmuijMap}
can be re-written as
\begin{subequations} \label{eq:rmuijMapDiff}
\beqa
    \tau^r_{i\la j}(t) &\approx& \tau^r_j(t). \\
    \rhat_{i\la j}(t) &\approx& \xhat_j(t) + \tau^r_{j}(t)\sum_{\ell \neq i} s_\ell(t)a_{\ell j}  \nonumber \\
        &=& \rhat_j(t) - \tau^r_j(t)a_{ij}s_i(t),
\eeqa
\end{subequations}
where in the approximations we have ignored terms of order $O(a_{ij}^2)$.
We can then simplify \eqref{eq:xhatijMap1} as
\beqa
     \lefteqn{\xhat_{i \la j}(\tp1) }\nonumber\\
     &\stackrel{(a)}{\approx}&
     \gin\left(\rhat_j(t)-a_{ij}s_i(t)\tau^r_j(t), q_j, \tau^r_j(t)\right) \nonumber\\
     &\stackrel{(b)}{\approx}&  \xhat_j(\tp1) -a_{ij}s_j(t)D_j(\tp1)
        \label{eq:xhatijD}
\eeqa
where (a) follows from substituting \eqref{eq:rmuijMapDiff} into \eqref{eq:xhatijMap1}
and (b) follows from a first-order approximation with the definitions
\begin{subequations}
\beqa
    \xhat_j(\tp1) &:=&  \gin\left(\rhat_{j}(t), q_j, \tau^r_{j}(t)\right)
        \label{eq:xhatjMap} \\
    D_j(\tp1) &:=& \tau^r_j(t)\frac{\partial}{\partial \rhat}
        \gin\left(\rhat_{j}(t), q_j, \tau^r_{j}(t)\right). \label{eq:dxjMap}
\eeqa
\end{subequations}
Now
\beqa
    \lefteqn{ D_j(\tp1) \stackrel{(a)}{\approx} \tau^r_j(t)
    \frac{\partial}{\partial \rhat}
        (\Gamma\fin(\cdot,q_j))\left(\rhat_{j}(t), \tau^r_{j}(t)\right) }\nonumber \\
    &\stackrel{(b)}{=}&  \frac{\tau^r_j(t)}{1 - \tau^r_j(t)\fin''(\xhat_j(\tp1),q_j) }
        \nonumber \\
    &\stackrel{(c)}{\approx}&  \left[ -\frac{\partial^2}{\partial x_j^2}
        \Delta_{i \la j}(\tp1,\xhat_j(\tp1)) \right]^{-1}
        \nonumber \\
    &\stackrel{(d)}{\approx}& \tau^x_j(\tp1), \label{eq:muxD}
\eeqa
where (a) follows from \eqref{eq:dxjMap} and by comparing
\eqref{eq:ginMap} to \eqref{eq:GammaDef};
(b) follows from \eqref{eq:GamDeriv};
(c) follows from \eqref{eq:DelLApp1} and
(d) follows from \eqref{eq:muxMapDel}.
Substituting \eqref{eq:muxD} into \eqref{eq:xhatijD} and \eqref{eq:zhatiMap},
we obtain
\beq \label{eq:zhatiMap1}
    \phat_i(t) = \sum_j a_{ij}\xhat_j(t) - \tau^p_i(t)\shat_i(\tm1),
\eeq
which agrees with the definition in \eqref{eq:zhati}.
In summary, we have shown that with the choice of $\gin(\cdot)$
and $\gout(\cdot)$ in \eqref{eq:ginMap} and \eqref{eq:goutMap},
the GAMP algorithm can be seen as a quadratic approximation
of max-sum loopy BP for MAP estimation.

\section{Sum-Product GAMP} \label{sec:bpMmse}

\subsection{Preliminary Lemma}
Our analysis will needs the following standard result.

\begin{lemma}  Consider a random variable $U$ with a conditional
probability density function of the form
\[
    p_{U|V}(u|v) := \frac{1}{Z(v)}\exp\left( \phi_0(u) + uv \right),
\]
where $Z(v)$ is a normalization constant (called the partition function).
Then,
\begin{subequations} \label{eq:logZRel}
\beqa
     \frac{\partial}{\partial v}\log Z(v) &=& \Exp(U|V=v) \label{eq:logZD1} \\
     \frac{\partial^2}{\partial v^2}\log Z(v)  &=&
        \frac{\partial}{\partial v}\Exp(U|V=v) \label{eq:logZD2} \\
         &=& \var(U|V=v). \label{eq:logZvar}
\eeqa
\end{subequations}
\end{lemma}
\begin{proof}  The relations are standard properties of exponential families
\cite{WainwrightJ:08}.
\end{proof}

\subsection{Sum-Product BP for MMSE Estimation}

The sum-product loopy BP algorithm for 
MMSE estimation is similar to max-sum algorithm 
in Appendix \ref{sec:bpMap}.
For the sum-product algorithm, the BP messages
can also be represented as functions
$\Delta_{i \ra j}(t,x_j)$ and $\Delta_{i \la j}(t,x_j)$.
However, these messages are to be interpreted
as estimates of the log likelihoods of the variables $x_j$,
conditioned on the system input and output vectors,
$\qbf$ and $\ybf$, and the transform matrix $\Abf$.

The updates for the messages are also slightly different
between the max-sum and sum-product algorithms.
For the sum-product algorithm, the output node update \eqref{eq:outBPMap} is
replaced by the update equation
\beq  \label{eq:outBPMmse}
    \Delta_{i \ra j}(t,x_j) = \log \Exp( p_{Y|Z}(y_i,z_i) | x_j )
    + \mbox{const},
\eeq
where the expectation is over the random variable $z_i = \abf_i^T\xbf$
with the components $x_r$ of $\xbf$ being independent with distributions
\beq \label{eq:pirMmse}
    p_{i \la r}(x_r) \propto \exp \Delta_{i \la r}(t,x_r).
\eeq
The constant term in \eqref{eq:outBPMmse} can be any term that does
not depend on $x_j$.
The sum-product input node update is identical to the max-sum update \eqref{eq:inBPMap}.
Similar to the max-sum estimation algorithm, sum-product BP is terminated
after some iterations.  The final estimate for the conditional distribution
of $x_j$ is given by
\[
    p_{j}(x_j) \propto \exp \Delta_{j}(t,x_j),
\]
where $\Delta_j(t,x_j)$ is given in \eqref{eq:inBPMapTot}.
From this conditional distribution, one can compute the conditional mean of $x_j$.

\subsection{GAMP Approximation}
We now show that with the $\gin(\cdot)$ in \eqref{eq:ginMmse}
and $\gout(\cdot)$ in \eqref{eq:goutMmse}, the GAMP algorithm
can heuristically be regarded as a Gaussian approximation of the
above updates.  The calculations are largely the same
as the approximations for max-sum BP in Appendix \ref{sec:bpMap},
so we will just point out the key differences.

We begin with the output node update \eqref{eq:outBPMmse}.
Let
\begin{subequations} \label{eq:xhatMuMmse}
\beqa
    \xhat_{j}(t) &:=& \Exp[ x_j | \Delta_{j}(t,\cdot) ], \label{eq:xhatMmseDelj} \\
    \xhat_{i \la j}(t) &:=& \Exp[ x_j | \Delta_{i \la j}(t,\cdot) ],
        \label{eq:xhatMmseDel} \\
    \tau^x_{j}(t) &:=& \var[ x_j | \Delta_{j}(t,\cdot) ]  \label{eq:muxMmseDelj}  \\
    \tau^x_{i \la j}(t) &:=& \var[ x_j | \Delta_{i \la j}(t,\cdot) ],
        \label{eq:muxMmseDel}
\eeqa
\end{subequations}
where we have used the notation $\Exp(g(x)|\Delta(\cdot))$ to mean the
expectation over a random variable $x$ with a probability density function
\beq \label{eq:pDel}
    p_\Delta(x) \propto \exp\left( \Delta(x) \right).
\eeq
Therefore, $\xhat_{i \la r}(t)$ and $\tau^x_{i \la r}(t)$ are the mean
and variance of the random variable $x_r$ with density \eqref{eq:pirMmse}.
Now, the expectation in \eqref{eq:outBPMmse} is over $z_i = \abf_i^T\xbf$
with the components $x_r$ being independent with probability density
\eqref{eq:pirMmse}.  So, for large $n$, the Central Limit Theorem suggests
that the conditional distribution of $z_i$ given $x_j$ should be approximately
Gaussian $z_i \sim {\cal N}(\phat_{i \la j}(t), \tau^p_{i \la j}(t))$,
where $\phat_{i \la j}(t)$ and $\tau^p_{i \la j}(t))$ are defined in
\eqref{eq:zmuijMap}.
Hence, $\Delta_{i \ra j}(t,x_j)$ can be approximated as
\beq \label{eq:DelRApp}
    \Delta_{i \ra j}(t,x_j) \approx H\left(\phat_{i \la j}(t),y_i,
        \tau^p_{i \la j}(t)\right),
\eeq
where
\beq \label{eq:Hmmse1}
    H(\phat,y,\tau^p) := \log\Exp\left[p_{Y|Z}(y|z)|\phat,\tau^p\right],
\eeq
and the expectation is over random variable $z \sim {\cal N}(\phat,\tau^p)$.
It is easily checked that
\beq \label{eq:Hmmse2}
    H(\phat,y,\tau^p) := \log p(z|\phat,y,\tau^p) + \mbox{const},
\eeq
where $p(z|\phat,y,\tau^p)$ is given in \eqref{eq:pFout} and
the constant term does not depend on $\phat$.

Now, identical to the argument in Appendix \ref{sec:bpApproxMap}, we can
define $\phat_i(t)$ as in \eqref{eq:zhatiMap}
and $\tau^p_i(t)$ as in \eqref{eq:mup} so that
$\Delta_{i \ra j}(t,x_j)$ can be approximated as \eqref{eq:DelRH2}.
Also, if we define $\gout(\phat,y,\tau^p)$ as in \eqref{eq:goutMap1},
and $\shat_i(t)$ and $\tau^s_i(t)$ as in \eqref{eq:outNL},
then  $\shat_i(t)$ and $\tau^s_i(t)$
are respectively the first and second-order
derivatives of $H(\phat_i(t),y_i,\tau^p_i(t))$ with respect to $\phat$.
Then, taking a second order approximation of \eqref{eq:DelRH2}
results in \eqref{eq:DelRAppMap}.

Also, using \eqref{eq:Hmmse2}, $\gout(\cdot)$ as defined in \eqref{eq:goutMap1}
agrees with the definition in \eqref{eq:goutMmseScore}.
To show that $\gout(\cdot)$ is also equivalent to \eqref{eq:goutMmse}
note that we can rewrite $H(\cdot)$ in \eqref{eq:Hmmse2} as
\[
    H(\phat,y,\tau^p) := \log \left[
        \frac{Z(\phat,y,\tau^p)}{Z_0(\phat,y,\tau^p)}\right],
\]
where
\beqan
    \lefteqn{ Z_0(\phat,y,\tau^p) := \int \exp\left[
        \frac{1}{2\tau^p}(2\phat z - z^2) \right]dz} \\
    \lefteqn{ Z(\phat,y,\tau^p) } \nonumber \\
    &:=& \int \exp\left[
        \fout(z,y) + \frac{1}{2\tau^p}(2\phat z - z^2) \right]dz.
\eeqan
Then the relations \eqref{eq:goutMmse} and \eqref{eq:goutMmseDeriv}
follow from the relations \eqref{eq:logZRel}.

We next consider the input node update \eqref{eq:inBPMap}.
Similar to Appendix \ref{sec:bpMap}, we can
substitute the approximation \eqref{eq:DelRAppMap} into
\eqref{eq:inBPMap} to obtain \eqref{eq:DelLApp1}
where $\rhat_{i \la j}(t)$ and $\tau^r_{i \la j}(t)$
are defined in \eqref{eq:rmuijMap}.

Using the approximation \eqref{eq:DelLApp1} and the definition of $\Fin(\cdot)$
in \eqref{eq:FinMap}, the probability distribution $p_\Delta(x_j)$
in \eqref{eq:pDel} with $\Delta = \Delta_{i \la j}(t,x_j)$ is approximately
\beqan
    \lefteqn{ p_{\Delta_{i \la j}(t,\cdot)}(x_j) }  \\
    &\approx& \frac{1}{Z}\exp \Fin(x_j,\rhat_{i \la j}(t), q_j, \tau^r_{i \la j}(t)),
\eeqan
where $Z$ is a normalization constant.
But, based on the form of $\Fin(\cdot)$ in \eqref{eq:FinMap}, this
distribution is identical to the conditional distribution of $X$
in $\theta^r(\tau^r)$ in \eqref{eq:thetaR} given $(\Rhat,Q) = (\rhat,q)$.
Therefore, if we define $\gin(\cdot)$ as in \eqref{eq:ginMmse},
it follows that $\xhat_{i \la j}(t)$ in \eqref{eq:xhatMmseDel} satisfies
\eqref{eq:xhatijMap1}.  The remainder of the proof now follows as in the proof
of Appendix \ref{sec:bpMmse}.

\section{Proof of \eqref{eq:dpzMatch}} \label{sec:alphaMatch}
Define
\beqa
    \lefteqn{ D(\Kbf^p) := \Exp\left[ \frac{\partial}{\partial \phat}
    \gout(t,h(Z,W),\Phat,\taubar^p(t)) \right]} \nonumber \\
    &-&
     \Exp\left[ \frac{\partial}{\partial z}
    \gout(t,h(Z,W),\Phat,\taubar^p(t)) \right], \label{eq:DKdef}
\eeqa
where the expectation is over $(Z,\Phat) \sim {\cal N}(0,\Kbf^p)$
and $W \sim p_W(w)$ independent of $(Z,\Phat)$.
We must show that when $\Kbf^p$ is of the form \eqref{eq:KpMatch}
and $\gout(\cdot)$ is given by \eqref{eq:goutMmse}, then
$D(\Kbf^p) = 0$.

To this end, let $\Qbf$ be the two-dimensional random vector
\beq \label{eq:QdefPf}
    \Qbf=[Z \ \ \Phat]^T \sim {\cal N}(0,\Kbf^p)
\eeq
and let $\Gbf \in \R^{1\x2}$ be the derivative
\[
    \Gbf = \Exp\left[ \frac{\partial}{\partial \qbf}\gout(t,h(Z,W),\Phat,\taubar^p(t))
        \right],
\]
which is the row vector whose two components are
the partial derivatives with respect to $z$ and $\phat$.
Thus, $D(\Kbf^p)$ in \eqref{eq:DKdef} can be re-written as
\beq \label{eq:DGstein}
    D(\Kbf^p) = \Gbf\left[ \begin{array}{c} 1 \\ -1 \end{array} \right].
\eeq
Also, using Stein's Lemma (Lemma \ref{lem:stein} below)
and the covariance \eqref{eq:QdefPf},
\beqan
   \lefteqn{ \Exp\left[ \Qbf \gout(t,h(Z,W),\Phat,\taubar^p(t)) \right] } \\
   &=& \Exp\left[ \Qbf\Qbf^T \right] \Gbf^T = \Kbf^p\Gbf'^T.
\eeqan
Applying this equality to \eqref{eq:DGstein} we obtain
\beq \label{eq:DStein}
    D(\Kbf^p) := \Exp\left[ \phi(Z,\Phat)\gout(t,h(Z,W),\Phat,\taubar^p(t)) \right],
\eeq
where $\phi(\cdot)$ is the function:
\[
    \phi(z,\phat) := \left[z \ \phat\right](\Kbf^p)^{-1}
        \left[ \begin{array}{c}1 \\ -1 \end{array} \right].
\]
When $\Kbf^p$ is of the form \eqref{eq:KpMatch}, then it is easily checked that
\beq \label{eq:phizp}
    \phi(z,\phat) = \frac{\phat}{\taubar^p(t)}.
\eeq
Therefore,
\beqa
    \lefteqn{    D(\Kbf^p) \stackrel{(a)}{=} \frac{1}{\taubar^p(t)}\Exp\left[
    \Phat\gout(t,h(Z,W),\Phat,\taubar^p(t)) \right] } \nonumber \\
    &\stackrel{(b)}{=}& \frac{1}{\taubar^p(t)}\Exp\left[
    \Phat \frac{\partial}{\partial \Phat} \log p_{Y|\Phat}(Y|\Phat) \right]
        \nonumber\\
    &\stackrel{(c)}{=}& \frac{1}{\taubar^p(t)}\Exp\left[
    \Phat \frac{\partial}{\partial \phat} \int p_{Y|\Phat}(y|\Phat)dy \right]
        \nonumber\\
    &=& \frac{1}{\taubar^p(t)}\Exp\left[
    \Phat \frac{\partial}{\partial \phat} (1) \right] = 0  \nonumber
\eeqa
where (a) follows from substituting \eqref{eq:phizp} into \eqref{eq:DStein};
(b) follows from \eqref{eq:goutMmseMatch};
(c) follows from taking the derivative of the logarithm.
This shows that $D(\Kbf^p) = 0$ and proves \eqref{eq:dpzMatch}.

\section{Proof Sketch for Claim~\ref{thm:SEGen}} \label{sec:SEGenPf}

As stated earlier, Bayati and Montanari in \cite{BayatiM:11}
already proved the result for scalar case when $n_b=n_d=1$.
Only very minor modifications are required for the vector-valued case,
so we will just provide a sketch of the key changes.

For the vector case, it is convenient to introduce the notation $\bbf(t)$
to denote the matrix with columns $\bbf_i(t)$:
\[
    \bbf(t) := \left[ \begin{array}{c}
        \bbf_1(t)^T \\ \vdots \\ \bbf_m(t)^T \end{array} \right] \in \R^{m \x n_b}.
\]
We can define $\ubf(t)$, $\vbf(t)$ and $\dbf(t)$ similarly.
Then, in analogy with the definitions in \cite{BayatiM:11}, we let
\beqan
    \xbf(t) &:=& \dbf(t) + \ubf(t)\xibf(t)^T \in \R^{n \x n_d} \\
    \ybf(t) &:=& \bbf(t) + \vbf(\tm1)\lambdabf(t)^T \in \R^{m \x n_b}
\eeqan
and define the matrices
\beqan
    \Xbf(t) &:=& \left[ \xbf(0) | \cdots | \xbf(\tm1) \right] \in \R^{n \x tn_d} \\
    \Ybf(t) &:=& \left[ \ybf(0) | \cdots | \ybf(\tm1) \right] \in \R^{m \x tn_b} \\
    \Ubf(t) &:=& \left[ \ubf(0) | \cdots | \ubf(\tm1) \right] \in \R^{n \x tn_b} \\
    \Vbf(t) &:=& \left[ \vbf(0) | \cdots | \vbf(\tm1) \right] \in \R^{m \x tn_d}.
\eeqan
The updates \eqref{eq:genAMPRec} can then be re-written
in matrix form as
\beq \label{eq:XYrec}
    \Xbf(t) = \Abf^T\Vbf(t), \ \
    \Ybf(t) = \Abf\Ubf(t).
\eeq

Next, also following the proof in \cite{BayatiM:11},
let $\vbf_{||}(t)$ be the projection of each column of
$\vbf(t)$ onto the column space of $\Vbf(t)$, and let $\vbf_{\perp}(t)$
be its orthogonal component, $\vbf_{\perp}(t) = \vbf(t)-\vbf_{||}(t)$.
Thus, we can write
\beq \label{eq:valpha}
    \vbf_{||}(t) = \sum_{i=0}^{t-1} \vbf(i)\alphabf_i(t),
\eeq
for matrices $\alphabf_i(t) \in \R^{n_d \x n_d}$.
Similarly, let $\ubf_{||}(t)$ and $\ubf_{\perp}(t)$ be the parallel and
orthogonal components of $\ubf(t)$ with respect to the column space
of $\Ubf(t)$ so that
\beq \label{eq:ubeta}
    \ubf_{||}(t) = \sum_{i=0}^{t-1} \ubf(i)\betabf_i(t),
\eeq
where $\betabf_i(t) \in \R^{n_b \x n_b}$.

Now let $\Ggothic(t_1,t_2)$ be the
sigma algebra generated by $\bbf(0),\ldots,\bbf(t_1)$,
$\vbf(0),\ldots,\vbf(t_1)$, $\dbf(0),\ldots,\dbf(t_2-1)$ and
$\ubf(0),\ldots,\ubf(t_2)$.  Also for any sigma-algebra, $\Ggothic$,
and random variables $X$ and $Y$, we will say that $X$ is
equal in distribution to $Y$ conditional on $\Ggothic$ if
$\Exp(\phi(X)Z) = \Exp(\phi(Y)Z)$ for any $Z$ that is $\Ggothic$-measurable.
In this case, we write $X|_{\Ggothic} \stackrel{d}{=} Y$.

With these definitions, the vector-valued analogy of the main technical lemma
\cite[Lemma 1]{BayatiM:11} can be stated as follows:

\begin{lemma} \label{lem:mainPf}  Under the assumptions of
Claim~\ref{thm:SEGen}, the following hold for all $t \geq 0$:
\begin{itemize}
\item[(a)]  The conditional distribution of $\dbf(\tp1)$ is given by
\beqan
    \lefteqn{
    \left. \dbf(t)\right|_{\Ggothic(t,t)} } \\
    &\stackrel{d}{=}&
      \sum_{i=0}^{\tm1} \dbf(i)\alphabf_i(t)
     + \widetilde{\Abf}^T\vbf_{\perp}(t)
    + \widetilde{\Ubf}(t)\overrightarrow{o}_t(1) \\
    \lefteqn{
    \left. \bbf(t)\right|_{\Ggothic(\tm1,t)} } \\
    &\stackrel{d}{=}&
      \sum_{i=0}^{\tm1} \bbf(i)\betabf_i(t)
     + \widetilde{\Abf}\ubf_{\perp}(t)
    + \widetilde{\Vbf}(t)\overrightarrow{o}_t(1)
\eeqan
where $\widetilde{\Abf}$ is an independent copy of $\Abf$,
and the matrices $\alphabf_i(t)$ and $\betabf_i(t)$
are the coefficients in the expansion
\eqref{eq:valpha} and \eqref{eq:ubeta}.
The matrix $\widetilde{\Ubf}(t)$ is such that its columns form
an orthogonal basis for the column space of $\Ubf(t)$ with
$\widetilde{\Ubf}(t)^T\widetilde{\Ubf}(t) = nI_{tn_d}$.
Similarly, $\widetilde{\Vbf}(t)$ is such that its columns form
an orthogonal basis for the column space of $\Vbf(t)$ with
$\widetilde{\Vbf}(t)^T\widetilde{\Vbf}(t) = mI_{tn_b}$.
The vector $\overrightarrow{o}_t(1)$ goes to zero almost surely.

\item[(b)] For any pseudo-Lipschitz functions
$\phi_d:\R^{n_d (t+1)} \x \Theta^u \arr \R$
and $\phi_b:\R^{n_b (t+1)} \x \Theta^v \arr \R$:
\beqan
   \lefteqn{ \lim_{n \arr \infty} \frac{1}{n} \sum_{j=1}^n
    \phi_d(\dbf_j(0), \cdots, \dbf_j(t), \theta^u_j)} \\
    &=& \Exp\left[ \phi_d(\Dbf(0), \ldots, \Dbf(t), \theta^u) \right] \\
   \lefteqn{ \lim_{n \arr \infty} \frac{1}{m} \sum_{i=1}^m
    \phi_b(\bbf_j(0), \cdots, \bbf_j(t), \theta^v_j)} \\
    &=& \Exp\left[ \phi_b(\Bbf(0), \ldots, \Bbf(t), \theta^v) \right]
\eeqan
where the expectations are over Gaussian
random vectors $(\Dbf(0),\ldots,\Dbf(t))$
and $(\Bbf(0),\ldots,\Bbf(t))$ and the random variables
$\theta^u$ and $\theta^v$ in the limit \eqref{eq:thetaUVLim}.
The variables $\theta^u$ and $\theta^v$ are independent
of $\Bbf(r)$ and $\Dbf(r)$ and the marginal distributions
of $\Bbf(r)$ and $\Dbf(r)$ are given by \eqref{eq:BDGauss}.

\item[(c)] For all $0 \leq r,s \leq t$, the following limit exists
hold are bounded and are non-random
\beqan
    \lefteqn{\lim_{n \arr \infty} \frac{1}{n}
    \sum_{j=1}^n \dbf_j(r)\dbf_j(s)^T} \\
    &=&
    \lim_{m \arr \infty} \frac{1}{m} \sum_{i=1}^m \vbf_i(r)\vbf_i(s)^T \\
    \lefteqn{\lim_{n \arr \infty} \frac{1}{m}
    \sum_{i=1}^m \bbf_i(r)\bbf_i(s)^T}\\
    &=&    \beta
    \lim_{m \arr \infty} \frac{1}{n} \sum_{i=1}^n \ubf_j(r)\ubf_j(s)^T.
    \hspace{0.5in}
\eeqan

\item[(d)] Suppose $\varphi_d: \R^{n_d} \x \Theta^u \arr \R$ and
$\varphi_b: \R^{n_b} \x \Theta^v \arr \R$ are almost
everywhere continuously differentiable with a bounded derivative with
respect to the first argument.
Then, for all all $0 \leq r,s \leq t$
the following limits exists are bounded and non-random:
\beqan
    \lefteqn{ \lim_{n \arr \infty} \frac{1}{n} \sum_{j=1}^n
        \dbf_j(r)\varphi_d(\dbf_j(s))^T } \\
        &=& \lim_{n \arr \infty}\frac{1}{n} \sum_{j=1}^n
        \dbf_j(r)\dbf_j(s)^T\Gbf_d(s)^T \\
    \lefteqn{\lim_{m \arr \infty} \frac{1}{m} \sum_{i=1}^m
        \bbf_i(r)\varphi_b(\bbf_j(s))^T } \\
        &=& \lim_{m \arr \infty} \frac{1}{m} \sum_{i=1}^m
        \bbf_i(r)\bbf_i(s)^T\Gbf_b(s)^T \hspace{0.5in}
\eeqan
where $\Gbf_d(s)$ and $\Gbf_b(s)$ are the empirical derivatives
\beqan
    \Gbf_d(s) &:=& \lim_{n \arr \infty} \frac{1}{n} \sum_{j=1}^n
        \frac{\partial}{\partial \dbf} \varphi_d( \dbf_j(s), \theta^u ) \\
    \Gbf_b(s) &:=& \lim_{n \arr \infty} \frac{1}{m} \sum_{i=1}^m
        \frac{\partial}{\partial \bbf} \varphi_b( \bbf_i(s), \theta^v ).
\eeqan

\item[(e)] For $\ell = k-1$, the following bounds hold almost surely
\beqan
    \lim_{n \arr \infty} \frac{1}{n} \sum_{j=1}^n \|\dbf_j(t)\|^{2\ell}
     &<& \infty \\
    \lim_{m \arr \infty} \frac{1}{m} \sum_{i=1}^m \|\bbf_i(t)\|^{2\ell}
     &<& \infty
\eeqan

\end{itemize}
\end{lemma}

This lemma is a verbatim copy of \cite[Lemma 1]{BayatiM:11}
with some minor changes for the vector-valued case.
Observe that
Claim~\ref{thm:SEGen} is a special case of part (b) of Lemma \ref{lem:mainPf}
by considering functions of the form
\beqan
    \phi_d(\dbf(0),\ldots,\dbf(t),\theta^u) &=& \phi_d(\dbf(t),\theta^u) \\
    \phi_b(\bbf(0),\ldots,\bbf(t),\theta^u) &=& \phi_b(\bbf(t),\theta^v).
\eeqan
That is, we consider functions that only depend on the most recent iteration
number.

The proof of Lemma \ref{lem:mainPf} also follows
follows almost identically to the proof of the analogous lemma in the
scalar case in \cite{BayatiM:11}.
The key idea in that proof is the following conditioning argument,
originally used by Bolthausen in \cite{Bolthausen:09}:
To evaluate the conditional distributions of $\dbf(t)$ with
respect to $\Ggothic(t,t)$ and $\bbf(t)$ with respect
to $\Ggothic(\tm1,t)$ in part (a)
of Lemma \ref{lem:mainPf}, one evaluates the corresponding conditional distribution
of the matrix $\Abf$.  But, this conditional distribution is precisely
identical to the distribution of $\Abf$ conditioned on \emph{linear}
constraints of the form \eqref{eq:XYrec}.  But, this distribution is just
the distribution of a Gaussian random vector conditioned on it lying on
an affine subspace.  That distribution has a simple expression as a
deterministic offset plus a projected Gaussian.  The detailed
expression are given in \cite[Lemma 6]{BayatiM:11} and the identical
equations can be used here.

The proof uses this conditioning principle along with
an induction argument along the iteration number $t$
and the statements (a) to (e) of the lemma.  We omit the details as they
are involved but follow with only minor changes from the original scalar
proof in \cite{BayatiM:11}.

The only one other non-trivial extension that is needed is the matrix
form of Stein's Lemma, which can be stated as:

\begin{lemma}[Stein's Lemma \cite{Stein:72}] \label{lem:stein}
 Suppose $\Zbf_1 \in \R^{n_1}$
and $\Zbf_2 \in \R^{n_2}$ are jointly Gaussian random vectors and
$\varphi: \R^{n_2} \arr \R^{n_3}$ is any function such that the expectations
\[
    \Gbf := \Exp\left[ \frac{\partial}{\partial \zbf_2}\varphi(\Zbf_2) \right]
    \in \R^{n_3 \x n_2}
\]
and
\[
    \Exp(\Zbf_1\varphi(\Zbf_2)^T) \in \R^{n_1 \x n_3}
\]
exists.
Then,
\[
    \Exp(\Zbf_1\varphi(\Zbf_2)^T) =
    \Exp\left((\Zbf_1-\overline{\Zbf}_1)(\Zbf_2-\overline{\Zbf}_2)^T \right)\Gbf^T,
\]
where $\overline{\Zbf}_i$ is the expectation of $\Zbf_i$.
\end{lemma}

Note that part (d) of Lemma \ref{lem:mainPf} is of the same form of this lemma.

\bibliographystyle{IEEEtran}
\bibliography{bibl}

\end{document}